\documentclass[11pt,a4paper]{article}

\usepackage[top=2.5cm,bottom=2.5cm,left=2.5cm,right=2.5cm]{geometry}
\usepackage{amsmath,amsthm,amssymb}
\usepackage{graphicx}

\newcommand{\scan}[1]{\text{Scan}\left(#1\right)}
\newcommand{\sort}[1]{\text{Sort}\left(#1\right)}

\newcommand{\OO}[1]{O\left(#1\right)}

\newcommand{\Om}[1]{\Omega\left(#1\right)}

\newcommand{\OM}[1]{\Omega\left(#1\right)}
\newcommand{\rep}{D.rep}

\newcommand{\deckey}{\textsc{DecreaseKey}}

\newcommand{\ins}{\textsc{Insert}}
\newcommand{\ext}{\textsc{Extract}}
\newcommand{\sea}{\textsc{Search}}
\newcommand{\bi}{\textsc{Batched-Insert}}

\newcommand{\be}{\textsc{Batched-ExtractMin}}
\newcommand{\re}{\textsc{Resolve}}
\newcommand{\flu}{\textsc{Flush-Up}}
\newcommand{\fld}{\textsc{Flush-Down}}
\newcommand{\spl}{\textsc{Split}}

\newcommand{\init}{\textsc{Initialize}}
\newcommand{\update}{\textsc{Update}}
\newcommand{\delete}{\textsc{Delete}}
\newcommand{\extmin}{\textsc{ExtractMin}}

\newcommand{\dx}{\left(D.x\right)}

\newcommand{\MM}{\frac{M}{B}}
\newcommand{\NN}{\frac{N}{B}}

\newtheorem{theorem}{Theorem}

\newtheorem{lemma}{Lemma}
\date{}

\title{Cache-Oblivious Priority Queues with Decrease-Key and Applications to Graph Algorithms} 
\author{John Iacono\\Department of Computer Science, Universit\'e Libre de Bruxelles, Belgium \and 
Riko Jacob\\Computer Science Department, IT University of Copenhagen, Denmark \and
Konstantinos Tsakalidis\\Department of Computer Science, University of Liverpool, United Kingdom}

\begin{document}

\maketitle

\begin{abstract}
	We present priority queues in the cache-oblivious external memory model with block size $B$ and main memory size $M$ that support on $N$ elements, operation \update\ (a combination of operations \ins~and \deckey) in $\OO{\frac{1}{B}\log_{\frac{\lambda}{B}} \frac{N}{B}}$ amortized I/Os and operations \extmin\ and \delete\ in $\OO{\lceil \frac{\lambda^{\varepsilon}}{B}  \log_{\frac{\lambda}{B}} \frac{N}{B} \rceil  \log_{\frac{\lambda}{B}} \frac{N}{B}}$ amortized I/Os, using $\OO{\frac{N}{B}\log_{\frac{\lambda}{B}} \frac{N}{B}}$ blocks, for a user-defined parameter $\lambda \in \left[2,   N \right]$ and any real $\varepsilon \in \left(0,1\right)$. In the cache-aware I/O model, \update\ takes optimal $\OO{\frac{1}{B}\log_{\MM} \frac{N}{B}}$ amortized I/Os by setting $\lambda = \OO{M}$. Previous I/O-efficient cache-oblivious and cache-aware priority queues either support these operations in $\OO{\frac{1}{B}\log_2 \frac{N}{B}}$ amortized I/Os [Chowdhury and Ramachandran, TALG 2018], [Brodal et al., SWAT 2004], [Kumar and Schwabe, SPDP 1996], or support only operations \ins, \delete\ and \extmin~in optimal I/Os, but without supporting \deckey\ [Arge et al., SICOMP 2007], [Fadel et al., TCS 1999].
	
	We also present buffered repository trees that support on a multi-set of $N$ elements, operation \ins\ in $\OO{\frac{1}{B}\log_{\frac{\lambda}{B}} \frac{N}{B}}$ I/Os and operation \ext\ on $K$ extracted elements in $\OO{\frac{\lambda^{\varepsilon}}{B} \log_{\frac{\lambda}{B}} \frac{N}{B} + \frac{K}{B}}$ amortized I/Os, using $\OO{\frac{N}{B}}$ blocks. Previous cache-oblivious and cache-aware results achieve $\OO{\frac{1}{B}\log_2 \frac{N}{B}}$ I/Os and $\OO{\log_2 \frac{N}{B} + \frac{K}{B}}$ I/Os, respectively [Arge et al., SICOMP '07], [Buchsbaum et al., SODA '00].
	
	In the cache-oblivious external memory model, for $\lambda = \OO{E/V}$, our results imply improved $\OO{\frac{E}{B}\log_{\frac{E}{V B}} \frac{E}{B}}$ I/Os for single-source shortest paths, depth-first search and breadth-first search algorithms on massive directed dense graphs $\left(V,E\right)$. Our algorithms are I/O-optimal for $E/V = \Omega \left(M\right)$. In the cache-aware setting, our algorithms are I/O-optimal for $\lambda = \OO{M}$.
\end{abstract}


\newpage

\section{Introduction}

Priority queues are fundamental data structures with numerous applications across computer science, most prominently in the design of efficient graph algorithms. They support the following operations on $N$ stored \emph{elements} of the type (\emph{key}, \emph{priority}), where ``key'' serves as an identifier and ``priority'' is a value from a total order:
\begin{itemize}
	\item \ins(element $e$): Insert element $e$ to the priority queue.
	\item \delete(key $k$): Remove all elements with key $k$ from the priority queue.
	\item element $e=$ \extmin(): Remove and return the element $e$ in the priority queue with the smallest priority.
	\item \deckey(element $(k,p)$): Given that an element with key $k$ and priority $p'$ is stored in the priority queue, if priority $p<p'$, replace the element's priority $p'$ with $p$. 
\end{itemize}

\noindent Operation \update(element $(k,p)$)~is the combination of operations \ins~and \deckey, that calls \ins($(k,p)$), if the priority queue does not contain any element with key $k$; or \deckey($(k,p)$), otherwise.  

We study the problem of designing cache-oblivious and cache-aware priority queues that support all these operations in external memory.  In the \emph{external memory model} (also known as the \emph{I/O model}) \cite{AV88} the amount of input data is assumed to be much larger than the main memory size $M$. Thus, the data is stored in an external memory device (i.e. disk) that is divided into consecutive blocks of size $B$ elements. Time complexity is measured in terms of \emph{I/O operations} (or \emph{I/Os}), namely block transfers from external to main memory and vice versa, while computation in main memory is considered to be free. Space complexity is measured in the number of blocks occupied by the input data in external memory. Algorithms and data structures in the I/O model are considered \emph{cache-aware}, since they are parameterized in terms of $M$ and $B$. In contrast, \emph{cache-oblivious} algorithms and data structures \cite{FLPR99} are oblivious to both these values, which allows them to be efficient along all levels of a memory hierarchy. I/O-optimally scanning and sorting $x$ consecutive elements in an array are commonly denoted to take  $\scan{x} = \OO{\frac{x}{B}}$ I/Os and $\sort{x} = \OO{\frac{x}{B}\log_\MM\frac{x}{B}}$ I/Os, respectively \cite{AV88,FLPR99}.



Priority queues are a basic component in several fundamental graph algorithms, including: 

\begin{itemize}
	\item The \emph{single-source shortest paths (SSSP)} algorithm on directed graphs with positively weighted edges, which computes the minimum edge-weight paths from a given source node to all other nodes in the graph.
	
	\item The \emph{depth-first search (DFS)} and \emph{breadth-first search (BFS)} algorithms on directed unweighted graphs, which number all nodes of the graph according to a depth-first or a breadth-first exploration traversal starting from a given source node, respectively.

\end{itemize}

\noindent Another necessary component for these algorithms are I/O-efficient \emph{buffered repository trees (BRTs)} \cite{BGVW00,ABDHM07,CR18}. They are used by external memory graph algorithms in order to confirm that a given node has already been visited by the algorithm. This avoids expensive random-access I/Os incurred by simulating internal memory methods in external memory. In particular, BRTs support the following operations on a stored multi-set of $N$ (key, value) elements, where ``key'' serves as an identifier and ``value'' is a value from a total order:

\begin{itemize}
	\item \ins(element $e$): Insert element $e$ to the BRT.
	\item element $e_{i} = $ \ext(key $k$): Remove and return all $K$ elements $e_i$ (for $i\in [1,K]$) in the BRT with key $k$.
\end{itemize}

\begin{table}
	\begin{center}
		\begin{tabular}{l l l l l}
			& \ins & \delete & \extmin & \deckey \\
			\cite{ABDHM07} & $\frac{1}{B}\log_{\frac{M}{B}}\NN$ & $\frac{1}{B}\log_{\frac{M}{B}}\NN$ & $\frac{1}{B}\log_{\frac{M}{B}}\NN$ & $-$ \\
			\cite{BFMZ04,CR18} & $\frac{1}{B}\log_{2}\NN$ & $\frac{1}{B}\log_{2}\NN$ & $\frac{1}{B}\log_{2}\NN$ & $\frac{1}{B}\log_{2}\NN$ \\
			New & $\frac{1}{B}\log_{\frac{\lambda}{B}} \frac{N}{B}$ &  $\lceil \frac{\lambda^{\varepsilon}}{B}  \log_{\frac{\lambda}{B}} \frac{N}{B} \rceil  \log_{\frac{\lambda}{B}} \frac{N}{B}$ &  $\lceil \frac{\lambda^{\varepsilon}}{B}  \log_{\frac{\lambda}{B}} \frac{N}{B} \rceil  \log_{\frac{\lambda}{B}} \frac{N}{B}$ & $\frac{1}{B}\log_{\frac{\lambda}{B}} \frac{N}{B}$\\\hline 			
			\cite{FJKT99,WY14} & $\frac{1}{B}\log_{\frac{M}{B}}\NN$ & $\frac{1}{B}\log_{\frac{M}{B}}\NN$ & $\frac{1}{B}\log_{\frac{M}{B}}\NN$ & $-$ \\
			\cite{KS96} & $\frac{1}{B}\log_{2}\NN$ & $\frac{1}{B}\log_{2}\NN$ & $\frac{1}{B}\log_{2}\NN$ & $\frac{1}{B}\log_{2}\NN$ \\
			\cite{JL19}$^*$ & $\frac{1}{B}\log_{\log N} \NN$ & $\frac{1}{B}\log_{\log N} \NN$ & $\frac{1}{B}\log_{\log N} \NN$ & $\frac{1}{B}\log_{\log N} \NN$ \\
			New & $\frac{1}{B}\log_{\MM}\NN$ & $\lceil\frac{M^{\varepsilon}}{B} \log_{\MM}\NN \rceil \log_{\MM}\NN$ & $\lceil\frac{M^{\varepsilon}}{B} \log_{\MM}\NN \rceil \log_{\MM}\NN$ & $\frac{1}{B}\log_{\MM}\NN$ 		\end{tabular}
	\end{center}
	\caption{Asymptotic amortized I/O-bounds of cache-oblivious and cache-aware priority queue operations (respectively, above and below the horizontal line) on $N$ elements, parameter $\lambda \in \left[ 2, N \right]$ and real $\varepsilon \in \left(0,1\right)$. $^*$Expected I/Os.}\label{tab:pq}
\end{table}

\subsection{Previous work} 
Designing efficient external memory priority queues able to support operation \deckey\ (or at least operation \update) has been a long-standing open problem \cite{KS96,FJKT99,WY14,ELY17,CR18,JL19}. I/O-efficient adaptations of the standard heap data structure (cache-aware \cite{FJKT99} and cache-oblivious \cite{ABDHM07}) or other cache-aware sorting-based approaches \cite{WY14}, despite achieving optimal base-$(M/B)$ logarithmic amortized I/O-complexity, fail to support operation \deckey. (Nevertheless, we use these priority queues as subroutines in our structure.) 
On the other hand, cache-aware adaptations of the tournament tree \cite{KS96} and cache-oblivious adaptations of the heap \cite{BFMZ04,CR18} data structures support all operations, albeit in not so efficient base-$2$ logarithmic amortized I/Os. Indeed, in the recent work of Eenberg, Larsen and Yu \cite{ELY17} it is shown that for a sequence of $N$ operations, any external-memory priority queue supporting \deckey\ must spend $\max$\{\ins, \delete, \extmin, \deckey\} $ = \Om{\frac{1}{B}\log_{\log N}B}$ amortized I/Os. Jiang and Larsen present matching randomized priority queues \cite{JL19}. 

The cache-aware BRTs introduced by Buchsbaum et al. \cite[Lemma 2.1]{BGVW00} and their cache-oblivious counterparts \cite{ABDHM07} support \ins\ in $\OO{\frac{1}{B}\log_{2} \NN}$ amortized I/Os and \ext\ on $K$ extracted elements in $\OO{\log_{2} \NN + \frac{K}{B}}$ amortized I/Os on a multi-set of $N$ stored elements. 

\subsection{Our contributions} 

We present cache-oblivious I/O-efficient priority queues that support on $N$ stored elements, operation \update\ in optimal  $\OO{\frac{1}{B}\log_{\frac{\lambda}{B}} \frac{N}{B}}$ amortized I/Os and operations \extmin\ and \delete\ in $\OO{\lceil \frac{\lambda^{\varepsilon}}{B}  \log_{\frac{\lambda}{B}} \frac{N}{B} \rceil  \log_{\frac{\lambda}{B}} \frac{N}{B}}$ I/Os, using $\OO{\frac{N}{B}\log_{\frac{\lambda}{B}} \frac{N}{B}}$ blocks, for a user-defined parameter $\lambda \in \left[ 2, N \right]$ and any real $\varepsilon \in \left(0,1\right)$. Our priority queues are the first to support operation \update\ (and thus \deckey and \ins) in a cache-oblivious setting in $o\left(\frac{1}{B} \log_{2} \NN \right)$ I/Os. This is the first I/O-optimal cache-aware \update\ bound, setting $\lambda = \OO{M}$. Our bounds improve on previous cache-aware \cite{KS96} and cache-oblivious \cite{BFMZ04,CR18} priority queues supporting \deckey, albeit at the expense of suboptimal I/O-efficiency for \extmin\ and \delete\ (respecting the lower bound of \cite{ELY17} for $\lambda =\OM{B\log_2 N}$). 
See Table \ref{tab:pq} for a comparison with previous external memory priority queues.

We also present cache-oblivious I/O-efficient BRTs that support on a multi-set of $N$ elements, operation \ins\ in $\OO{\frac{1}{B}\log_{\frac{\lambda}{B}} \frac{N}{B}}$ amortized I/Os and operation \ext\ on $K$ extracted elements in $\OO{\frac{\lambda^{\varepsilon}}{B}  \log_{\frac{\lambda}{B}} \frac{N}{B} + \frac{K}{B}}$ amortized I/Os. Our bounds also hold in a cache-aware setting for $\lambda = \OO{M}$. Previous cache-oblivious and cache-aware I/O-bounds are $\OO{\frac{1}{B}\log_2 \frac{N}{B}}$ and $\OO{\log_2 \frac{N}{B} + \frac{K}{B}}$, respectively \cite{ABDHM07,BGVW00}.

Combining our BRTs with our priority queues, for cache-oblivous external memory SSSP, DFS and BFS algorithms, we achieve $\OO{\frac{V^{\frac{1}{1+\alpha}} E^{\frac{\alpha}{1+\alpha}}}{B}\log^2_{\frac{E}{VB}} \frac{E}{B} + V\log_{\frac{E}{VB}} \frac{E}{B} +  \frac{E}{B} \log_{\frac{E}{VB}} \frac{E}{B}}$ I/Os on graphs with $V$ nodes and $E$ directed edges, setting $\lambda = \OO{E/V}$. 
For cache-aware extrenal memory SSSP, DFS and BFS algorithms, we achieve $\OO{V \frac{M^{\frac{\alpha}{1+\alpha}}}{B}\log^2_{\frac{M}{B}} \frac{E}{B} + V\log_{\frac{M}{B}} \frac{E}{B} +  \frac{E}{B} \log_{\frac{M}{B}} \frac{E}{B}}$ I/Os, setting $\lambda = \OO{M}$. This compares to previous cache-oblivious and cache-aware graph algorithms that take $\OO{\left(V+\frac{E}{B}\right)\log_{2} E}$ I/Os for directed SSSP \cite{KS96,V01,CR18} and that take $\OO{\left(V+\frac{E}{B}\right)\log_{2} \frac{V}{B} + \frac{E}{B}\log_{\frac{M}{B}}\frac{E}{B}}$ I/Os for directed DFS and BFS \cite{BGVW00,ABDHM07}. 
Our cache-oblivious and cache-aware bounds are I/O-optimal, for dense graphs with $E/V = \Om{M}$ and with $E = \Om{V^{1+\varepsilon}}$ and $V \!\! = \Om{M}$, respectively.

\subsection{Our approach} 

The main component of our priority queues is the $x$-\emph{treap}, a recursive structure inspired by similar cache-oblivious $x$-box \cite{BDFILM10} and cache-aware hashing data structures \cite{IP12} that solve the dynamic dictionary problem in external memory (respectively, under predecessor and membership queries on a dynamic set of keys). To solve the priority queue problem, we adapt this recursive scheme to also handle priorities, inspired by the cache-oblivious priority queues of Brodal et al.~\cite{BFMZ04} and of Chowdhury and Ramachandran \cite{CR04,CR18} that support \update, yet in suboptimal~I/Os. We hope that the discussion below provides the intuition to follow the full details in the sequel. 

Previous cache-oblivious priority queues~\cite{BFMZ04,CR04,CR18} are based on a simple idea. Their basic structure has a logarithmic number of levels, where level $i$ has two arrays, or buffers, of size roughly $2^i$. These buffers are called the \emph{front} and \emph{rear} buffers. They contain key-priority pairs or a key-delete message (described later). The idea is that the front buffers are sorted, with everything in the $i$-th front buffer having smaller priorities than everything in the $(i+1)$-th front buffer. The items in the rear buffers do not have this rigorous ordering, but instead must be larger than the items in the rear buffer at the smaller levels. When an \update\ operation occurs, the key-priority pair gets placed in the first rear buffer; when a \extmin\ operation occurs, the key-priority pair with the smallest priority is removed from the first front buffer. Every time a level-$i$ buffer gets too full or empty relative to its target size of~$2^i$, this is fixed by moving things up or down as needed, and by moving things from the rear to the front buffer if that respects the ordering of items in the front buffer. This resolution of problems is done efficiently using a scan of the affected and neighbouring levels. Thus, looking in a simplified manner at the lifetime of an \update d item, it will be inserted in the smallest rear buffer, be pushed down to larger rear buffers as they overflow, be moved from a rear buffer to a front buffer once it has gone down to a level where its priority is compatible with those in the corresponding front buffer, then moves up from the front buffer to smaller front buffers as they underflow, and is finally removed from the smallest front buffer during an \extmin. Thus, during its lifetime, it could be moved from one level to another a total of $\OO{\log_2 \frac{N}{B}}$ times at an I/O-cost of $\OO{\frac{1}{B}}$ per level, for a total cost of $\OO{\frac{1}{B} \log_2 \frac{N}{B}}$ I/Os. One detail is that when an item moves from a rear to a front buffer, we want to make sure that no items in larger levels with the same key and larger priority are ever removed. This is done through special delete messages, which stay in the rear buffers and percolate down, removing any key-priority pairs with the given key that they encounter in their buffer or the corresponding front buffer.

The problem with this approach is that the base-2 logarithm seems unavoidable, with the simple idea of a geometrically increasing buffer size. So here instead we use the more complicated recursion introduced with the cache-oblivious $x$-box  structure \cite{BDFILM10} and also used in the cache-aware hashing data structures \cite{IP12}. In its simplest form, used for a dictionary, an $x$-box has three buffers: top, middle and bottom (respectively of approximate size $x$, $x^{1.5}$ and $x^2$), as well as $\sqrt{x}$ recursive \emph{upper-level} $\sqrt{x}$-boxes (ordered logically between the top and middle buffers) and $x$ recursive \emph{lower-level} $\sqrt{x}$-boxes (ordered logically between the middle and bottom buffers). Data in each buffer is sorted, and all keys in a given recursive buffer are smaller than all keys in subsequent recursive buffers in the same level (upper or lower). There is no enforced order among keys in different buffers or in a recursive upper- or lower-level $\sqrt{x}$-box. The key feature of this construction is that the top/middle/bottom buffers have the same size as the neighbouring recursive buffers: the top buffer has size $x$, the top buffers of the upper-level recursive $\sqrt{x}$-boxes have total size $x$; the middle buffer, sum of the bottom buffers of the upper-level, and sum of the top buffers of the lower-level recursive structures all have size $x^{1.5}$; the sum of the bottom buffers of the lower-level recursive structures and the bottom buffer both have size $x^2$. Therefore, when for example a top buffer overflows, it can be fixed by moving excess items to the top buffers of the top recursive substructures. In a simplified view with only insertions, as buffers overflow, an item over its lifetime will percolate from the top buffer to the upper-level substructures, to the middle buffer, to the lower-level substructures, and to the bottom buffer, with each overflow handled only using scans. Assuming a base case of size $M$, there will be $\OO{\log_M N}$ times that an item will move from one buffer to another and an equal number of times that an item will pass through a base case. One major advantage of this recursive approach, is that an item will pass though a small base case not just once at the top of the structure, as in the previous paragraph, but many times. 

We combine these ideas to form the $x$-treap, described at a high level as follows: Everywhere an $x$-box has a buffer, we replace it with a front and a rear buffer storing key-priority pairs. The order used by the $x$-box is imposed on the keys, not the priorities. The order imposed on priorities in the previous cache-oblivious priority queues are carried over and imposed on the priorities in different levels of the $x$-treap; this is aided by the fact that the buffers in the $x$-treap form a DAG, thus the buffers where items with a given key can appear, form a natural total order. Hence, this forms a treap-like arrangement where we use the keys for order in one dimension and priorities for order in the other. We invoke separate trivial base case structure at a size smaller than a fixed value, e.g. the main memory size in a cache-aware setting; it stores items in no particular order and thus supports fast insertion of items when a neighbouring buffer adds them ($\OO{\frac{1}{B}}$), but slow ($\OO{M^\epsilon}$ amortized) removal of items with small priorities to fix the underflow of a front buffer above. In its typical hypothetical lifetime, an item will be inserted at the top in the rear buffer, percolate down $\OO{\log_{\frac{M}{B}} \frac{N}{B}}$ levels and base cases at a cost of $\OO{\frac{1}{B}}$ amortized each, move over to a front buffer, then percolate up $\OO{\log_{\frac{M}{B}} \frac{N}{B}}$ levels at a cost of $\OO{\frac{M^\epsilon}{B}}$ amortized each. Thus, the total amortized cost for an item that is eventually removed by an \extmin\ is $\OO{\frac{M^\epsilon}{B}\log_{\frac{M}{B}} \frac{N}{B}}$. 

However, we want the amortized cost for an item that is inserted via \update\ to be much faster than this, i.e. $\OO{\frac{1}{B}\log_{\frac{M}{B}} \frac{N}{B}}$. This requires additional observations and tricks. The first is that, unlike Brodal et al., we do not use delete-type messages that percolate down to eliminate items with larger than minimum priority in order to prevent their removal from \extmin. Instead, we adopt a much simpler approach, and use a hash table to keep track of all keys that have been removed by an \extmin, and when an \extmin\ returns a key that has been seen before, it is discarded and \extmin\ is repeated. The second trick is to simply ensure that each buffer has at most one item with each key (and remove key-priority pairs other than the one with the minimum priority among those with the same key in the buffer). This has the effect that if there are a total of $u$ updates performed on a key before it is removed by an \extmin, the total cost will involve up to $\OO{u \log_{\frac{M}{B}} \frac{N}{B}}$ percolations down at a cost of $\OO{\frac{1}{B}}$, but only $\OO{\log^2_{\frac{M}{B}} \frac{N}{B}}$ percolations up at a cost of $\OO{\frac{M^\epsilon}{B}}$ amortized each. After the \extmin, some items may still remain in the structure and will be discarded when removed by \extmin. However, due to the no-duplicates-per-level property there will only be $\OO{\log_{\frac{M}{B}} N}$ such items (called \emph{ghosts}) each of which will incur an amortized cost of at most $\OO{\lceil \frac{M^\epsilon}{B}\log_{\frac{M}{B}} \frac{N}{B}\rceil}$, where the ceiling accounts for accessing the hash table. Thus the total amortized cost for the lifetime of the $u$ \update s and one \extmin\ involving a single key is 
$O\left(\frac{u}{B}\log_{\frac{M}{B}} \frac{N}{B} +\lceil \frac{M^\epsilon}{B}\log_{\frac{M}{B}} \frac{N}{B}\rceil \log_{\frac{M}{B}} \frac{N}{B} \right).$ 
This cost can be apportioned in the amortized sense by having the \extmin\ cost $\OO{\lceil \frac{M^\epsilon}{B}\log_{\frac{M}{B}} \frac{N}{B}\rceil \log_{\frac{M}{B}} \frac{N}{B}}$ amortized~I/Os and each update cost
$\OO{\frac{1}{B}\log_{\frac{M}{B}} \frac{N}{B} }$ amortized~I/Os, assuming that the treap finishes in an empty state and more importantly that no item can be \update d after it has been \extmin 'd. 

The details that implement these rough ideas consume the rest of the paper. One complication that eludes the above discussion is that items don't just percolate down and then up; they could move up and down repeatedly and this can be handled through an appropriate potential function. The various layers of complexity needed for the $x$-treap recursion combined with the front/rear buffer idea, various types of over/underflows of buffers, a special base case, having the middle and bottom buffers be of size $x^{1+\frac{\alpha}{2}}$ and $x^{1+\alpha}$ for a suitable parameter $\alpha$ rather than $x^{1.5}$ and $x^2$ as described above, and a duplicate-catching hash table, result in a complex structure with an involved potential analysis, but that follows naturally from the above high-level description.

\section{Cache-oblivious $x$-Treap}
\label{sec:xtreap}

Given real parameter $\alpha \in (0,1]$ and \emph{key range} $\left[k_{\min}, k_{\max}\right)\subseteq \mathbb{R}$, an $x$-\emph{treap} $D$ stores a set of at most $2\left(D.x\right)^{1+\alpha} $ \emph{elements} $\left(\ast,k,p\right)$ associated with a \emph{key} $k\in \left[D.k_{\min}, D.k_{\max}\right)$ and a \emph{priority} $p$ from a totally ordered set. $D$ represents a set $\rep$ of pairs (key, priority), such that a particular key $k$ contained in $D$ is represented to have the smallest priority $p$ of any element with key $k$ stored in $D$, unless an element with key $k$ and a smaller priority has been removed from the structure. In particular, we call the key and priority \emph{represented}, when the pair (key, priority) $\in \rep$. A \emph{representative} element contains a represented key and its represented priority.  
Formally, we define:
$$
\rep := \bigcup_{\{k|\left(k , p \right)\in D\}}
\left \{ \left(k,\min_p \left(k, p \right) \in D \right)\right \}
$$

\noindent The proposed representation scheme works under the assumption that a key that is not represented by the structure anymore, cannot become represented again. In other words, a key returned by operation \extmin\ cannot be \ins ed to the structure again. 
The following \emph{interface operations} are supported:


\begin{itemize}
	\item \bi $\left(D, e_1, e_2 , \ldots, e_{b}\right)$: For constant $c\in \left(0,\frac{1}{3}\right]$, insert $b \leq c\cdot D.x$ elements $e_1, e_2 , \ldots, e_{b}$  to $D$, given they are key-sorted with keys  $e_i.k \in \left[D.k_{\min}, D.k_{\max}\right), i\in \left[1,b\right]$. 
	
	\bi\ adds the pairs $\left(e_i.k, e_i.p\right)$ to $\rep$ with key $e_i.k$ that is not contained in $D$ already. \bi\ decreases the priority of a represented key $e_i.k$ to $e_i.p$, if its represented priority is larger than $e_i.p$ before the operation. More formally, let $X_{new}$ contain the inserted pairs $\left(e_i.k, e_i.p\right)$ with $e_i.k\notin \rep$. Let $X_{old}$ contain the pairs in $\rep$ with an inserted key, but with larger priority than the inserted one, and let $X_{dec}$ contain these inserted pairs. After \bi, a new $x$-treap $D'$ is created where $D'.rep  = \rep \cup X_{new} \cup X_{dec} \backslash X_{old}$.
	
	%
	%
	\item \be $\left(D\right)$: For constant $c\in \left(0,\frac{1}{4}\right]$, remove and return the at most $c\cdot D.x$ elements $\left(k , p \right)$ with the smallest priorities in $D$. 
	
	\be\ removes the pairs $X_{\min}$ from $\rep$ with the at most $c\cdot D.x$ smallest priorities. Let $X_{key}$ contain the pairs in $D$ with keys in $X_{\min}$. After \be, a new $x$-treap $D'$ is created where $D'.rep  = \rep \backslash X_{\min} \backslash X_{key}$.
\end{itemize}

\begin{theorem}\label{thm:xtreap}
	An $x$-treap $D$ supports operation \be\ in $\OO{\lambda^{\frac{\alpha}{1+\alpha}}\frac{1+\alpha}{B} \log_{\lambda} D.x}$ amortized I/Os per element and operation \bi\ in  $\OO{\frac{1+\alpha}{B}\log_{\lambda} D.x}$ amortized I/Os per element, using $\OO{\frac{\dx^{1+\alpha}}{B}\log_{\lambda} D.x}$ blocks, for some $\lambda \in \left[2, N \right]$ and for any real $\alpha\in (0,1]$. 
\end{theorem}

The lemmata in this section prove the above theorem. 
The structure is recursive. The base case is described separately in Subsection \ref{ssec:base}. The base case structure is used when $D.x \leq c' \lambda^{\frac{1}{1+\alpha}}$ (for an appropriately chosen constant $c'>0$). Thus assuming $D.x> c' \lambda^{\frac{1}{1+\alpha}}$, we define an $x$-treap to contain three \emph{buffers} (which are arrays that store elements) and many $\sqrt{x}$-treaps (called \emph{subtreaps}). Specifically, the \emph{top}, \emph{middle} and \emph{bottom} buffers have \emph{sizes} $D.x,  \dx^{1+ \frac{\alpha}{2}}$ and $\dx^{1+\alpha}$, respectively. Each buffer is divided in the middle into a \emph{front} and a \emph{rear} \emph{buffer}. The subtreaps are divided into the \emph{upper} and the \emph{lower level} that contain at most $\frac{1}{4}\dx^{\frac{1}{2}}$ and $\frac{1}{4}\dx^{\frac{1+\alpha}{2}}$ subtreaps, respectively.  
Let $|b|$ denote the \emph{size} of a buffer $b$.
We define the \emph{capacity} of an $x$-treap $D$ to be the maximum number of elements it can contain,  which is $D.x + \frac{5}{4}\dx^{1+\frac{\alpha}{2}} + \frac{5}{4}\dx^{1+\alpha} < 2\dx^{1+\alpha}$.

We define a partial order ($\preceq$) using the terminology ``above/below'' among the buffers of an $x$-treap and all of the buffers in recursive subtreaps or base case structures. In this order we have top buffer $\preceq$ upper level recursive subtreaps
$\preceq$ middle buffer
$\preceq$ lower level recursive subtreaps
$\preceq$ bottom buffer. 

Along with all buffers of $D$, we store several additional pieces of bookkeeping information: a counter with the total number of elements stored in $D$ and an index indicating which subtreap is stored in which space in memory.

\begin{figure}
	\begin{center}
		\includegraphics[scale=0.4]{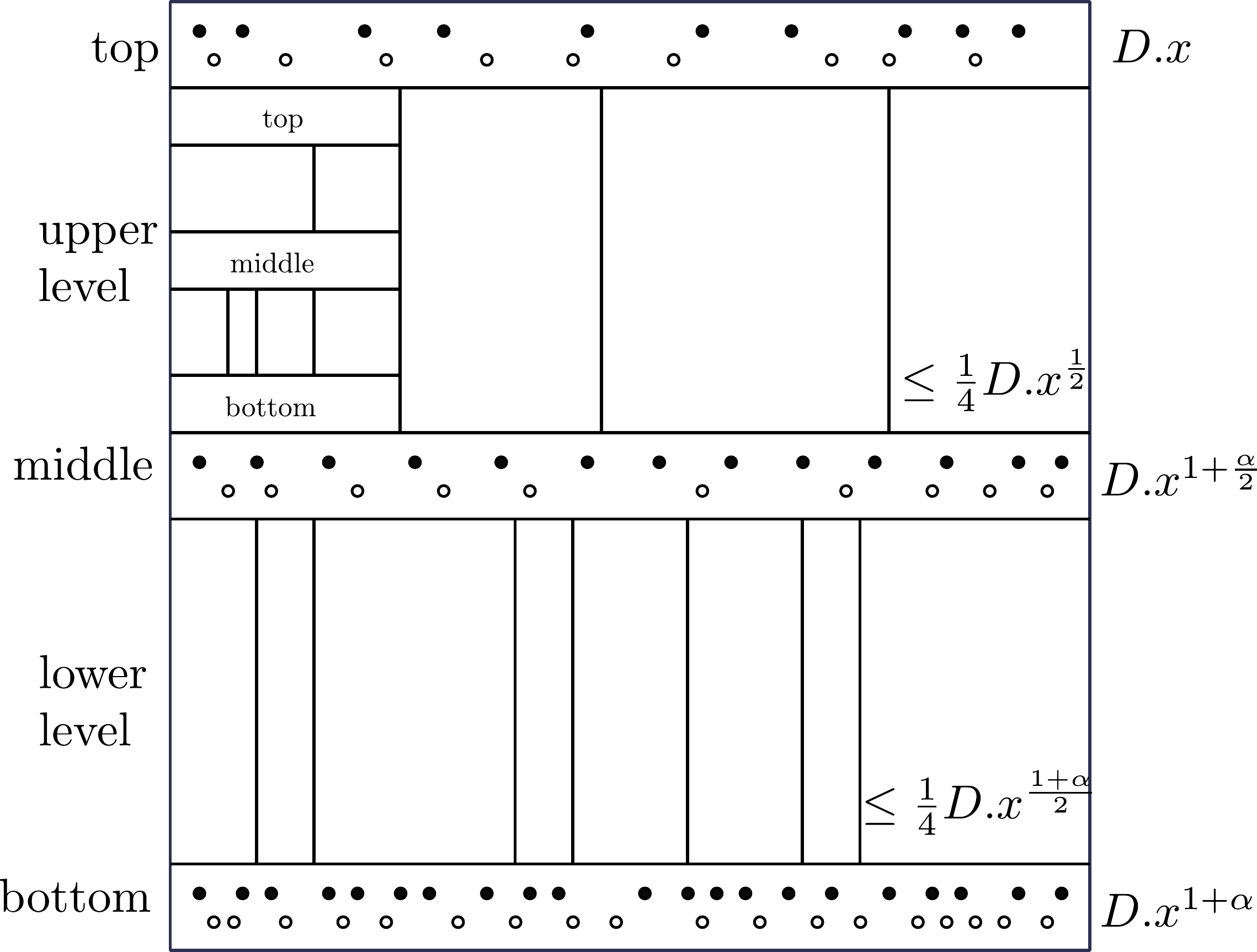}
	\end{center}
	\caption{Overview of an $x$-treap $D$ on ``key'' $\times$ ``partial order''  space. Black/white dots represent elements in the front/rear buffers, respectively. All buffers are resolved. Buffer sizes and a level's maximum number of subtreaps appear on the right-hand side.}
\end{figure}

\subsection{Invariants} An $x$-treap $D$ maintains the following invariants with respect to every one of its top/middle/bottom buffers $b$. The invariants hold after the execution of each interface operation, but may be violated during the execution. They allow changes to $D$ that do not change $\rep$.

\begin{enumerate}
	\item \label{inv:type} The front and rear buffers of $b$ store elements sorted by key and left-justified.
	
	\item \label{inv:lprio} The front buffer's elements' priorities are smaller than the rear buffer's elements' priorities.
	
	\item \label{inv:prio} The front buffer's elements' priorities are smaller than all elements' priorities in buffers below $b$ in $D$.
	
	\item \label{inv:key} For a top or middle buffer $b$ with key range  $\left[ b.k_{\min}, b.k_{\max} \right)$, the $r$ upper or lower subtreaps $D_i, i\in \{1,r\}$, respectively, have distinct key ranges $\left[D_i.k_{\min} , D_i.k_{\max}\right)$, such that $b.k_{\min} = D_1.k_{\min}< D_1.k_{\max} = D_2.k_{\min} < \ldots < D_r.k_{\max} = b.k_{\max}$. 
	
	\item \label{inv:empty} If the middle/bottom buffer $b$ is not empty, then at least one upper/lower subtreap is  not empty, respectively.
\end{enumerate}


\subsection{Auxiliary operations}

The operations \bi\ and \be\ make use of the following \emph{auxiliary operations}: 

\begin{itemize}

	\item Operation \re$\left(D,b\right)$. We say that a buffer $b$ is \emph{resolved}, when the front and rear buffers contain elements with pairs (key,priority) $\left(k,p\right)$, such that $k$ is a represented key, and when the front buffer contains those elements with smallest priorities in the buffer. To resolve $b$, operation \re\
	assigns to the elements with represented keys, the key's minimum priority stored in $b$. Also, it removes any elements with non-represented keys from $b$. \re\ restores Invariant \ref{inv:lprio} in $b$, when it is temporarily violated by other (interface or auxiliary) operations that call it. 
	
	\item Operation \init$\left(D, e_1, e_2 , \ldots, e_{b}\right)$ distributes to a new $x$-treap $D$, $\frac{1}{4}\dx \leq b \leq\frac{1}{2} \dx^{1+\alpha}$ elements $e_i, i\in [1,b]$ from a temporary array (divided in the middle into a front and a rear array, respecting Invariants~\ref{inv:type} and \ref{inv:lprio}).
	
	\item Operation \flu$\left(D\right)$ ensures that the front top buffer of $D$ contains at least $\frac{1}{4}D.x$ elements (unless all buffers of $D$ contains less elements altogether, in which case \flu\ moves them all to the top front buffer of $D$). By Invariants \ref{inv:lprio} and \ref{inv:prio}, these are the elements in $D$ with smallest priority.
	
	\item Operation \fld$\left(D\right)$ is called by \bi\ on an $x$-treap $D$ whose bottom buffer contains between $\frac{1}{2}\dx^{1+\alpha}$ and $\dx^{1+\alpha}$ elements. 	It moves to a new temporary array, at least $\frac{1}{6}\dx^{1+\alpha}$  and at most $\frac{2}{3}\dx^{1+\alpha}$ elements from the bottom buffer of $D$. It ensures that the largest priority elements are removed from $D$.
	
	\item Operation \spl$\left(D \right)$ is called by \bi\ on an $x$-treap $D$ that contains between $\frac{1}{2}\dx^{1+\alpha}$ and $\dx^{1+\alpha}$ elements. It moves to a new temporary (front and rear) array, the at most $\frac{1}{3}\dx^{1+\alpha}$
	elements with largest keys in $D$.
	
\end{itemize}

\subsubsection{Resolving a buffer} 

\paragraph{Algorithm.} 

Auxiliary operation \re\ on a buffer $b$ of an $x$-treap $D$ is called by the auxiliary operation \flu\ and by the interface operation \bi. It makes use of two temporary auxiliary arrays of size~$|b|$. \re$\left(D, b\right)$ is implemented as following: 

\begin{enumerate} 
	\item \label{res:0} Determine the maximum priority $p_{\max}$ in the front buffer (by a scan). Return, if the front buffer is empty.
	\item \label{res:1} 2-way merge the elements in the front and rear buffers into a temporary array (by simultaneous scans in increasing key-order). Empty the front and rear buffers.
	\item \label{res:2} Determine the representative elements in the temporary array (by a scan) and write them to a second temporary array (by another scan): specifically for each key, write only the element with the smallest priority to the secondary temporary array. 
	
	\item \label{res:3} Scan the second temporary array, writing the elements with priority smaller or equal to $p_{\max}$ to the front buffer, and with priority larger than $p_{\max}$ to the rear buffer. Discard the temporary arrays.
	
	\item \label{res:4} Update the counter of $D$.
\end{enumerate}

\paragraph{Correctness.} 

After calling \re\ on a buffer $b$, elements from $b$ are allowed to be moved to other buffers, since Invariants \ref{inv:type} and \ref{inv:lprio} are maintained. This is because after Steps \ref{res:1}, \ref{res:2} and \ref{res:3}, the front and rear buffers of $b$ contain a representative element for every represented key in $b$ separated by priority $p_{\max}$ (computed in Step \ref{res:0}). Step \ref{res:4} accounts for the elements ignored in Step \ref{res:2}. 

See Figure \ref{fig:resol} for an illustration of the operation.

\begin{figure}
	
	\begin{center}
		\includegraphics[scale=0.5]{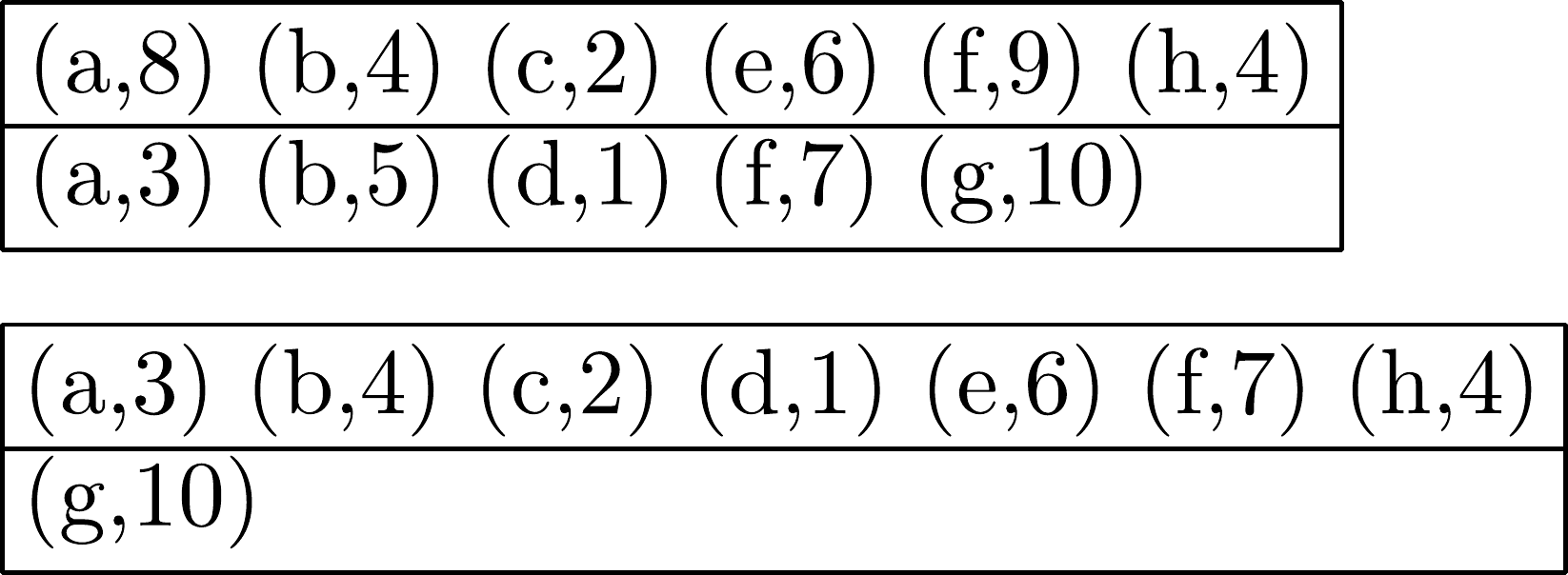}
	\end{center}
	\caption{A buffer before and after operation \re\ (respectively, above and below).}
	\label{fig:resol}
\end{figure}

\subsubsection{Initializing an $x$-treap}

\paragraph{Algorithm.} 

Auxiliary operation \init\ is called by the auxiliary operation \flu\ and by the interface operation \bi. It allocates an empty $x$-treap $D$ and distributes the $b\in \left[\frac{1}{4}\dx, \frac{1}{2}\dx^{1+\alpha} \right] $ elements $e_i$, $i\in\left[1,b\right]$  from a temporary key-sorted array (divided in the middle into a front and a rear array) to the buffers of $D$. \init$\left(D, e_1. \ldots, e_b\right)$ is implemented as following:

\begin{enumerate}
	\item \label{init:0} Create a new $x$-treap $D$ and move all elements in the temporary rear array to the rear bottom buffer of $D$.
	
	\item \label{init:1} Find the  $\left(\frac{1}{2} \dx^{1+\alpha}\right)$-th smallest priority in the temporary front array (by an order-statistics algorithm \cite{BFPRT73}) and move all elements in the array with larger priority to the front bottom buffer of $D$. 
	
	\item \label{init:2} Find the  $\left(\frac{1}{2} \dx\right)$-th smallest priority in the temporary front array and move all elements in the array with smaller priority to the front top buffer of $D$. 
	
	\item \label{init:3} Find the  $\left(\frac{1}{2} \dx\right)$-th smallest priority in the temporary front array and until the maximum number of upper level subtreaps has been reached:	\init\ a new upper subtreap with the $\frac{1}{2}\dx^{\frac{1+\alpha}{2}}$ key-next elements with smaller priority. 	 
	
	\item \label{init:4} Find the  $\left(\frac{1}{2} \dx^{1+\frac{\alpha}{2}}\right)$-th smallest priority in the temporary front array and move all elements in the array with smaller priority to the front middle buffer of $D$. 
	
	\item \label{init:5} Find the  $\left(\frac{1}{2} \dx^{1+\frac{\alpha}{2}}\right)$-th smallest priority in the temporary front array and until the maximum number of lower level subtreaps has been reached: \init\ a new lower subtreap with the $\frac{1}{2}\dx^{\frac{1+\alpha}{2}}$ key-next elements with smaller priority. 	 
	
	\item \label{init:6} Discard the temporary array and update the counters of $D$. 

\end{enumerate}

\paragraph{Correctness.} 

Operation \init\ moves the elements from the temporary array to a new $x$-treap in the following sequence: bottom rear buffer, top front buffer, upper subtreaps' front buffers, middle front buffer and lower subtreaps' front buffers, bottom front buffer. The recursive calls in Steps \ref{init:3} and \ref{init:5} ensure that the temporary array empties. All invariants are maintained. 

See Figure \ref{fig:init} for an illustration of the operation.

\begin{figure}
	\begin{center}
		\includegraphics[scale=0.5]{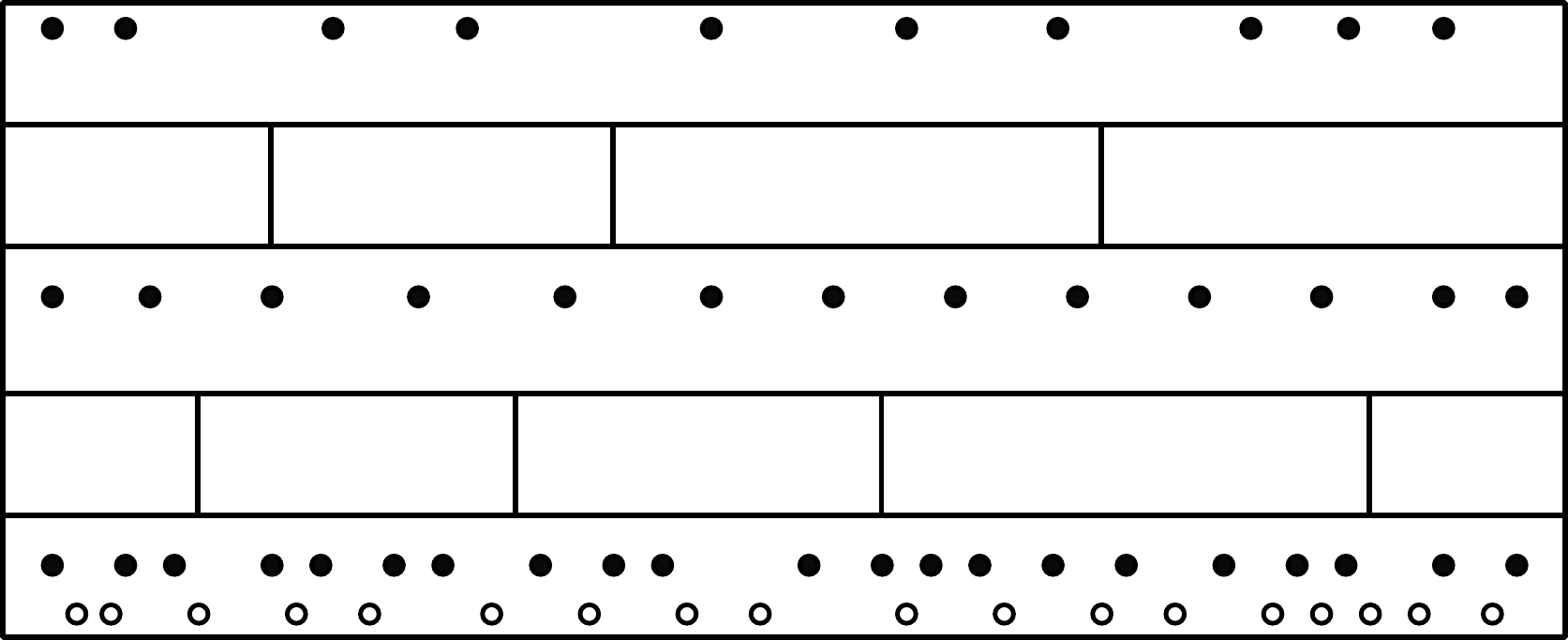}
	\end{center}
	\caption{A new $x$-treap after operation \init.}
	\label{fig:init}
\end{figure}

\subsubsection{Flushing up an $x$-treap}

\paragraph{Algorithm.} 

Auxiliary operation \flu\ on an $x$-treap $D$ is called only by the interface operation \be. It is implemented by means of the recursive subroutine \flu$\left(D, b\right)$ that also takes as argument a top or middle buffer $b$ of $D$ and moves to its front buffer, the elements with the (at most) $\frac{1}{4}|b|$ smallest represented priorities among the representative elements stored inside and below $b$ in $D$. The operation makes use of a temporary priority queue that supports only operations \ins\ and \extmin\ \cite{ABDHM07} (The structure can be easily modified to ). 
For a bottom buffer $b$, a non-recursive subroutine \flu$\left(D, b\right)$ simply calls \re\ on $b$. \flu$\left(D, b\right)$ is implemented as following:

\begin{enumerate}
	
	\item \label{flu:0} \re\ $b$ and \flu\ the top buffers of the subtreaps immediately below $b$.
	
	\item \label{flu:1} If the front buffer of $b$ contains $k <\frac{1}{4}|b|$ elements: Allocate a temporary array of size $|b|$ and a temporary priority queue $Q$. For every subtreap immediately below $b$: Remove all elements from its front top buffer and \ins\ them to $Q$ (by simultaneous scans). While the temporary array contains no more than $\frac{1}{4}|b| - k$ elements, do:
	
	\subitem \ref{flu:1}.1. \extmin\ one element $e$ from $Q$ and write it in the temporary array.  
	
	\subitem \ref{flu:1}.2. If $Q$ contains no more elements from the subtreap $D'$ that contained $e$: \flu\ the top buffer of~$D'$, remove all its elements and \ins\ them to $Q$.
	
	\subitem \ref{flu:1}.3. If $Q$ is empty and the temporary buffer contains $k'< \frac{1}{4}|b| - k$ elements: \flu\ the buffer $b'$ immediately below $b$ in $D$, find the  $\left(\frac{1}{4}|b|-k-k'\right)$-th smallest priority in the front buffer of $b'$ (by an external memory order-statistics algorithm \cite{BFPRT73}) and move its elements with smaller priority to the temporary array. Left-justify $b'$.
	
	\item \label{flu:2} Sort the elements in the temporary array by key. 2-way merge into the front buffer of $b$, the elements in the front buffer of $b$ and the temporary array (by simultaneous scans in increasing key-order). Discard the temporary array. 
	
	\item \label{flu:3} If $Q$ is not empty: \extmin\ all elements from $Q$ into a new temporary array, sort the array by key, move the elements left-justified back to the front top buffers of the subtreaps they were taken from (by simultaneous scans in increasing key-order), update the subtreaps' counters and discard the array. 
	
	\item \label{flu:4} Discard $Q$, update the counters of $D$ and remove all empty subtreaps immediately below $b$ (i.e. whose counter is $0$). 
	
	\item \label{flu:5} If there are no subtreaps immediately below $b$ and the front buffer $b'$ immediately below $b$ is not empty: Find the  $\left(\frac{1}{4}|b|^{\frac{1+\alpha}{2}}\right)$-th smallest priority in the front buffer of $b'$ (by an external memory order-statistics algorithm \cite{BFPRT73}), move the elements with smaller priority to a new temporary front array (by a scan), left-justify the front buffer of $b$ and \init\ a new subtreap with the elements in the array. Discard the temporary array.
	
	
\end{enumerate}

\paragraph{Correctness.} 

Operation \flu\ allows for accessing the representative elements with smallest represented priorities in $D$ by only accessing its front top buffer. Invariants \ref{inv:lprio} and \ref{inv:prio} imply that the next-larger represented priorities with respect to the front top buffer's maximum represented priority are stored in the upper subtreaps' front top buffers and in turn the next-larger ones are stored in the front middle buffer. Similarly, this holds between the middle buffer with respect to the lower subtreaps and the bottom buffer. 

The subroutine \flu$\left(D, b\right)$ respects this sequence when it moves elements with minimum represented priorities from front buffers to the front buffer of $b$. Specifically, Step \ref{flu:0} ensures that the representative elements in $b$ with priority smaller than the $\left(\frac{1}{4}|b|\right)$-th largest priority in $b$ are stored in its front buffer. It also ensures this recursively for the front top buffers of the subtreaps immediately below $b$. If the front buffer of $b$ contains less than $\frac{1}{4}|b|$ such elements, Step \ref{flu:1} attempts to move into a temporary array enough elements from below $b$ in $D$. The temporary priority queue is used (at Steps \ref{flu:1} and \ref{flu:1}.1) to ensure that indeed the smallest-priority representative elements are moved, first from the subtreaps (Step \ref{flu:1}.2) and, if not enough elements have been moved, from the buffer immediately below $b$ in $D$ (at Step \ref{flu:1}.3). At most two key-sorted runs are created in the temporary array (one from the subtreaps and one from the buffer immediatelly below $b$) which are merged to the front buffer of $b$ by Step \ref{flu:2}, while maintaining Invariants  \ref{inv:type} and \ref{inv:key}. Step \ref{flu:3} revokes the effects of the temporary priority queue, allowing it to be discarded at Step \ref{flu:4}, which also accounts for the moved elements. It also removes any empty subtreaps, maybe violating Invariant \ref{inv:empty} that is restored by Step \ref{flu:5}. 

See Figure \ref{fig:flu} for an illustration of the operation.

\begin{figure}
	
	\begin{center}
		\includegraphics[scale=0.5]{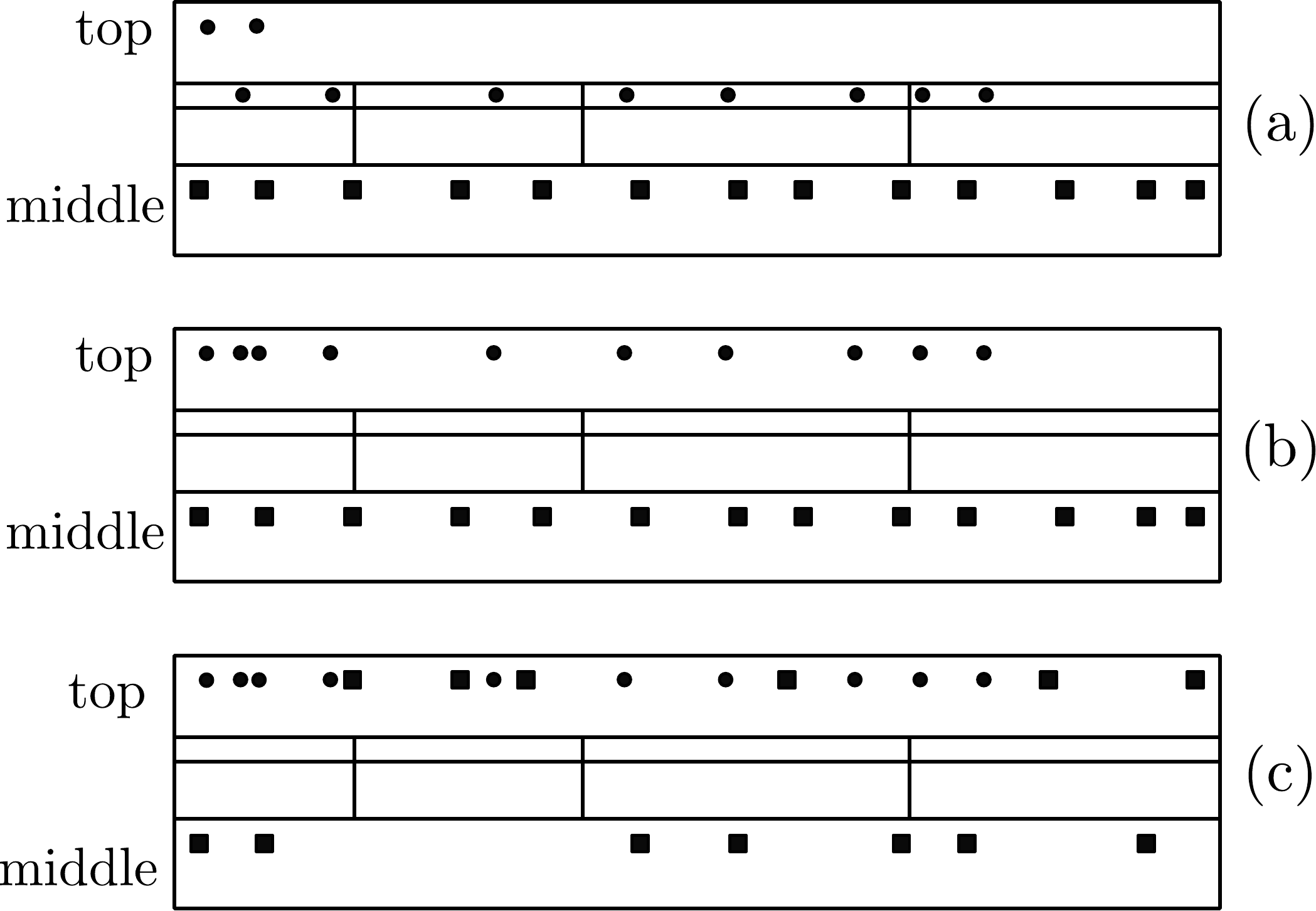}
	\end{center}
	\caption{The top/middle buffers and the top upper level buffers, (a) before operation \flu, (b) after Step \ref{flu:4} and (c) Step~\ref{flu:5}.}
	\label{fig:flu}
\end{figure}

\subsubsection{Flushing down an $x$-treap}
\paragraph{Algorithm.} 

Auxiliary operation \fld\ on an $x$-treap $D$ is 
called only by 
the interface operation \bi\ and returns a new temporary key-sorted array with at most  $\frac{2}{3}\dx^{1+\alpha}$ elements. \fld$\left(D\right)$ is implemented as following: 

\begin{enumerate}
	
	\item \label{fld:0} Move all elements from the bottom rear buffer of $D$ the the temporary rear array (by a scan).
	
	\item \label{fld:1} If Step \ref{fld:0} did not move more than $\left(\frac{1}{6}\dx^{1+\alpha}\right)$ elements: Find the $\left(\frac{1}{3}\dx^{1+\alpha}\right)$-th smallest priority in the bottom front buffer of $D$ (by an external memory order-statistics algorithm \cite{BFPRT73}). Move all elements in the bottom front buffer with larger priority to the temporary front array and left-justify the bottom front buffer (by a scan).
	
	
	\item \label{fld:2} 2-merge the two runs created by Steps \ref{fld:0} and \ref{fld:1} in the temporary array (by a scan).
	\item \label{fld:3} Update the counter of $D$.
\end{enumerate}

\paragraph{Correctness.} 

Operation \fld\ leaves the bottom rear buffer of $D$ empty and the bottom front buffer with at most $\frac{1}{3}\dx^{1+\alpha}$ elements. By Invariants \ref{inv:lprio} and \ref{inv:prio}, the largest priority elements of $D$ are in the bottom rear buffer and they are removed at Steps \ref{fld:0}. However, if they don't account for a constant fraction of $D$'s size, Step 
\ref{fld:1} removes such a fraction from the bottom front buffer, which contains the  next-smaller elements. Invariant \ref{inv:type} is maintained. Step \ref{fld:3} accounts for the removed elements. 

See Figure \ref{fig:fld} for an illustration of the operation.

\begin{figure}
	\begin{center}
		\includegraphics[scale=0.5]{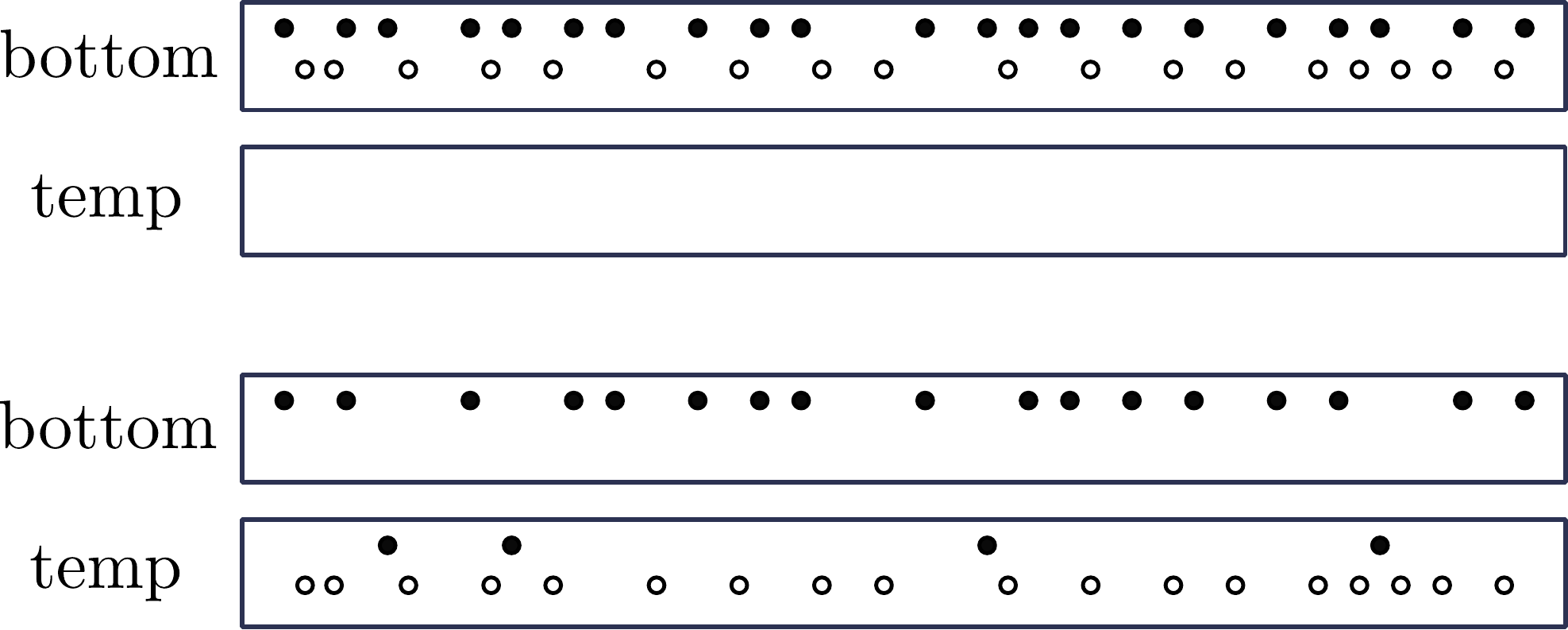}
	\end{center}
	\caption{The bottom buffer and temporary array before and after operation \fld\ (respectively, above and below).}
	\label{fig:fld}
\end{figure}

\subsubsection{Splitting an $x$-treap}

\paragraph{Algorithm.} 

Auxiliary operation \spl\  is called only by the interface operation \bi. It moves to a new temporary key-sorted array (divided in the middle into a front and a rear array) at most $\frac{1}{4}|b_i|$ key-larger elements from all buffers $b_i$ in $D$. \spl$\left(D\right)$ is implemented as following:

\begin{enumerate}
	
	\item \label{spl:0} Find the $\left(\frac{1}{4}\dx^{1+\alpha}\right)$-th smallest key in the front top/middle/bottom buffer of $D$ (by a scan). (Let this key be $k$ for the front buffer and $k'$ for the middle buffer.) Respectively, move all elements in the front buffers with larger key to three new front auxiliary arrays.

	\item \label{spl:1} Repeat Step \ref{spl:0} with respect to rear buffers.
	
	\item \label{spl:2} Update the counter of $D$.
	
	\item \label{spl:3} \spl\ the upper subtreap whose key range contains $k$. \spl\ the lower subtreap whose key range contains~$k'$. 
	
	\item \label{spl:4} Merge the elements in all front/rear auxiliary arrays to a new front/rear temporary array, respectively. 
	Discard all auxiliary arrays.  
\end{enumerate}

\paragraph{Correctness.} 

Operation \spl\ leaves all (front and rear) buffers of $D$ half-full. Since it operates on $x$-treaps that are more than half-full, whose bottom buffers contain a constant fraction of the total number of elements in the $x$-treap, the execution of Steps \ref{spl:0}, \ref{spl:1} and \ref{spl:3} moves at most $\frac{1}{2}\dx^{1+\alpha}$ elements to the temporary array. Step \ref{spl:2} accounts for the removed elements. All invariants are maintained. 

\subsection{Base case} \label{ssec:base} 
The $x$-treap is a recursive structure. When the $x$-treap stores few elements, we use simple arrays to support the interface operations and operation \flu. 

\begin{lemma}\label{lem:array}
	An $\OO{\lambda^{\frac{1}{1+\alpha}}}$-treap supports operation \bi\ in $\OO{1/B}$ amortized I/Os per element and operations \be\ and \flu\ in $\scan{\lambda^{\frac{\alpha}{1+\alpha}}}$ amortized I/Os per element, for some $\lambda \in \left[2, N \right]$ and for any real $\alpha\in (0,1]$. 
\end{lemma}

\begin{proof}
	For a universal constant $c_0>0$ and a constant parameter $c'<c_0^{\frac{1}{\alpha}+1}$, we allocate an array of size $\left(c'\left(\lambda^{\frac{1}{1+\alpha}}\right)\right)^{\frac{\alpha}{1+\alpha}} \leq c_0 M$ and divide it in the middle into a front and a rear buffer that store elements and maintain only Invariants \ref{inv:type} and \ref{inv:lprio}. 
	
	To implement \bi\ on at most $\frac{c'}{2}\lambda^{\frac{1}{1+\alpha}}$ elements, we simply add them to the rear buffer and update the counter. This costs $\OO{\frac{\lambda^{\frac{1}{1+\alpha}}}{B}/\frac{1}{2}\lambda^{\frac{1}{1+\alpha}}} = \OO{\frac{1}{B}}$ I/Os amortized per added element, since we only scan the part of the rear buffer where the elements are being added to. 
	
	To implement \be\ on at most $\frac{c'}{2}\lambda^{\frac{1}{1+\alpha}}$ extracted elements, we \re\ the array (as implemented for Theorem \ref{thm:xtreap}), remove and return all elements in the front buffer, and update the counter. By Lemma \ref{lem:res} (proven  in Subsection \ref{ssec:ana}) this costs $\OO{\frac{\lambda}{B}/\frac{1}{2}\lambda^{\frac{1}{1+\alpha}}} = \OO{\frac{\lambda^{\frac{\alpha}{1+\alpha}}}{B}}$ I/Os amortized per extracted element. 
	
	\flu\ is implemented like \be\ with the difference that the returned elements are not removed from the array. 
\end{proof}

\subsection{Interface operations} 

We proceed to the description of the interface operations supported by an $x$-treap. 

\subsubsection{Inserting elements to an $x$-treap}

\paragraph{Algorithm.} 

Interface operation \bi\ on an $x$-treap $D$ is implemented by means of the recursive subroutine \bi$\left(D,e_1, \ldots, e_{c\cdot |b|}, b\right)$ that also takes as argument a top or middle buffer $b$ of $D$ and inserts $c\cdot |b|$ elements $e_1, \ldots, e_{c\cdot |b|}$ (contained in a temporary array) inside and below $b$ in $D$, for constant $c\in \left(0,\frac{1}{3}\right]$.  For a bottom buffer $b$, a non-recursive subroutine \bi$\left(D, b\right)$ simply executes Step \ref{bi:0} below and discards the temporary array. \bi$\left(D,e_1, \ldots, e_{c\cdot |b|}, b\right)$ is implemented as following:

\begin{enumerate}
	
	\item \label{bi:0} If $D.x > c' \lambda^{\frac{1}{1+\alpha}} + c|b|$:
	
	\subitem \ref{bi:0}.1. 2-way merge into the temporary array, the elements in the temporary array and in the rear buffer of~$b$ (by simultaneous scans in increasing key-order). \re\ $b$ considering the temporary array as the rear buffer of $b$.
	
	\subitem \ref{bi:0}.2. Implicitly partition the front buffer of $b$ and the temporary array by the key ranges of the subtreaps immediately below $b$. Consider the subtreaps in increasing key-order by reading the index of $D$. For every key range (associated with subtreap $D'$) that contains at least $\frac{1}{3}\dx^{\frac{1}{2}}$ elements in either the front buffer of~$b$ or the temporary array: While the key range in the front buffer of $b$ and in the temporary array contains at most $\frac{2}{3}\dx^{\frac{1}{2}}$ elements, do:
	
	\subsubitem \ref{bi:0}.2.1. Find the $\left(\frac{2}{3}\dx^{\frac{1}{2}}\right)$-th smallest priority within the key range in the front buffer of $b$ and in the temporary array (by an external memory order-statistics algorithm \cite{BFPRT73}) and move the elements in the key range with larger priority to a new auxiliary array (by simultaneous scans in increasing key-order).
	
	\subsubitem \ref{bi:0}.2.2. If the counter of $D'$ plus the auxiliary array's size does not exceed the capacity of $D'$: \bi\ the elements in the auxiliary array to the top buffer of $D'$. Discard the auxiliary array. 
	
	\subsubitem \ref{bi:0}.2.3. Else, if there are fewer than the maximum allowed number of subtreaps in the level immediately below $b$: \spl\ $D'$. Let $k$ be the smallest key in the array returned by \spl\ (determined by a constant number of random accesses to the leftmost elements in the returned front/rear array). Move the elements in the auxiliary array with key smaller than $k$ to a new temporary array (by a scan), \bi\ these elements to $D'$ and discard this temporary array. 2-way merge the remaining elements in the auxiliary array into the returned rear array and discard the auxiliary array. \init\ a new subtreap with the elements in the returned array. Discard the returned array. 
	
	\subsubitem \ref{bi:0}.2.4. Else, \fld\ all subtreaps immediately below $b$, which writes them to many returned arrays. 2-way merge into a new temporary array, all elements in $b$ and in all returned arrays (by simultaneous scans in increasing key-order). (When the scan on a subtreap's temporary array is over, determine the subtreap with the key-next elements in the level by reading the index of $D$.) \bi\ the elements in the new temporary array to the buffer $b'$ immediately below $b$. Discard the new temporary array and all returned arrays.
	
	\subitem \ref{bi:0}.3. Discard the temporary array and update the counter of $D$. 
	\item \label{bi:1} Else if $D.x \leq c' \lambda^{\frac{1}{1+\alpha}} + c|b|$: \bi\ the elements to the base case structure.
	
\end{enumerate}

\paragraph{Correctness.} 

Operation \bi\ accommodates the insertion of at most $\frac{1}{3}|b|$ elements by allocating recursively extra space within $D$. Step \ref{bi:0} considers the recursive structure. Specifically, Step \ref{bi:0}.1 allows for moving representative elements from $b$ and inserted elements by resolving $b$ with respect to the temporary array. Step \ref{bi:0}.2 identifies the subtreaps immediately below $b$ (repeatedly in increasing key-order) whose associated key range contains too many keys (stored in $b$ and the temporary array, but not in the considered subtreap) and attempts to move the largest-priority elements within this key range into the subtreap. Step \ref{bi:0}.2.1 identifies the at most  $\frac{1}{3}|b|^{\frac{1}{2}}$ elements to be moved and Step \ref{bi:0}.2.2 recursively inserts them to the subtreap. However if the subtreap is full, Step \ref{bi:0}.2.3 splits it into two subtreaps with enough space. Nonetheless if the level cannot contain a new subtreap, Step \ref{bi:0}.2.4 essentially moves all elements in $b$ and in all subtreaps immediately below $b$, to the buffer $b'$ immediately below $b$. Step \ref{bi:0}.3 accounts for the number of inserted elements and the changes in number of subtreaps. Step \ref{bi:1} allows for recursing down to the base case.

\subsubsection{Extracting minimum-priority elements from an $x$-treap}

\paragraph{Algorithm.} 

Interface operation \be\ on an $x$-treap $D$ is implemented as following:

\begin{enumerate}
	\item \label{be:0} If $D.x > c' \lambda^{\frac{1}{1+\alpha}}$:
	\subitem \ref{be:0}.1 If the front top buffer contains less than $\frac{1}{4} D.x$ elements: \flu~the top buffer.
	\subitem \ref{be:0}.2 Remove and return all the elements $\left(e_i.k, e_i.p\right)$ from the front top buffer.
	\subitem \ref{be:0}.3 Update the counter of $D$.
	\item \label{be:1} Else if $D.x \leq c' \lambda^{\frac{1}{1+\alpha}}$: \be\ the base case structure. 
	
\end{enumerate}

\paragraph{Correctness.} 

Operation \be\ considers only the top buffer of $D$. Step \ref{be:0}.1 ensures that there are enough minimum-priority representative elements in the front top buffer of $D$ to be extracted by Step \ref{be:0}.2. Step \ref{be:0}.3 accounts for the extracted elements and Step \ref{be:1} for the base case. All Invariants are maintained.

\subsection{Analysis}
\label{ssec:ana}

\begin{lemma}\label{lem:card}
	An $x$-treap $D$ has $\OO{\log_{\lambda} D.x}$ levels and occupies $\OO{\dx^{1+\alpha}\log_{\frac{\lambda}{B}} D.x}$ blocks. 
\end{lemma}

\begin{proof}
	We number the levels of the structure sequentially, following the defined ``above/below'' order from top to bottom, where a base case structure counts for one level. Hence, the total number of levels is given by $L\left(D.x\right) = 3 + 2 L\left(\dx^{\frac{1}{2}}\right) $ with $L\left(c\lambda^{\frac{1}{1+\alpha}}\right) = 1$, which solves to the stated bound. Since, \re\ leaves in the operated buffer  at most one element with a given key, the space bound follows.
\end{proof}

\begin{lemma}\label{lem:ra}
	By the tall-cache assumption, scanning the buffers of an $x$-treap $D$ and randomly accessing $\OO{\dx^{\frac{1+\alpha}{2}}}$ subtreaps takes  $\scan{\dx^{1+\alpha}}$ I/Os, for any real $\alpha \in (0,1]$.
\end{lemma}

\begin{proof}
	We show that $\OO{\frac{\dx^{1+\alpha}}{B} + \dx^{\frac{1+\alpha}{2}}} = \OO{\frac{\dx^{1+\alpha}}{B}}$. Indeed this holds when $\dx^{\frac{1+\alpha}{2}} = \OO{\frac{\dx^{1+\alpha}}{B}}$. Otherwise $\dx^{\frac{1+\alpha}{2}} = \Om{\frac{\dx^{1+\alpha}}{B}} \Rightarrow \dx^{1+\alpha} = \OO{B^2} $ and by the tall-cache assumption that $M\geq B^2$, we get that	$D.x = \OO{M^{\frac{1}{1+\alpha}}}=\OO{M}$. Hence $D$ fits into main memory and thus randomly accessing its subtreaps incurs no I/Os, meaning that the I/O-complexity is dominated by $\OO{\frac{\dx^{1+\alpha}}{B}}$.
\end{proof}

\subsubsection{Amortization}

A buffer $b_i$ at level $i\leq h := \OO{\log_{\frac{\lambda}{B}} D.x}$ with $b_f$ elements in the front buffer and $b_r$ elements in the rear buffer has potential $\Phi (b_i) := \Phi_{f}(b_i) + \Phi_r(b_i)$, such that (for constants $\varepsilon := \frac{\alpha}{1+\alpha}$ and $c_0\geq 1$):

\begin{itemize}
	\item $\Phi_f (b_i) =\begin{cases}
	0, & \text{if $\frac{1}{4}|b_i|\le b_f \le \frac{1}{3}|b_i|$},\\
	\frac{c_0}{B}  \lambda^{\varepsilon} \cdot  \left(\frac{|b_i|}{4} - b_f\right) \cdot \left( h - i \right), & \text{if $b_f < \frac{1}{4}|b_i|$},\\
	\frac{c_0}{B} \cdot \left( b_f - \frac{|b_i|}{3} \right)  \cdot \left( h - i \right), & \text{if $b_f > \frac{1}{3}|b_i|$},\\
	\end{cases}$
	\item
	$\Phi_r (b_i) =\begin{cases}
	2 \frac{c_0}{B} \cdot \left(b_r - \frac{|b_i|}{2}\right) \cdot \left( h - i \right), & \hspace{11pt}\text{if $b_r > 0$}.
	\end{cases}$
\end{itemize}

In general, a particular element will be added to a rear buffer and will be moved down the levels of the structure over rear buffers by operation \fld. A \re\ operation will move the element from the rear to the front buffer, if it is a representative element. From this point, it will be moved up the levels over front buffers by operation \flu. If it is not representative, it will either get discarded by \re (when there is an element with the same key and with smaller priority in the same buffer) or it will keep going down the structure. Since \re\ leaves only one element per key at the level it operates, $\OO{\log_{\frac{\lambda}{B}} D.x}$ elements with the same key (i.e. at most one per level) will remain in the structure after the extraction of the representative element for this key. 

The $\lambda^\varepsilon$-factor accounts for the extra cost of \flu\ and \be, the $\left(h-i\right)$-factor allows for moving elements up or down a level by \flu\ and \fld\ and the $2$-factor accounts for moving elements from the rear to the front buffer. 

\begin{lemma}\label{lem:res}
	\re\ on a buffer $b_i$ takes $\scan{|b_i|}+\OO{1}$ amortized I/Os.
\end{lemma}

\begin{proof}
	All steps of operation \re\ are implemented by a constant number of scans over buffers of size at most $|b_i|$. 
	Since elements can only be added to the front buffer and only be removed from the rear buffer, the maximum difference in potential occurs, when $\frac{|b_i|}{2} - \frac{|b_i|}{3} = \frac{|b_i|}{6}$ elements are moved from a full rear buffer to a front buffer that contains $\frac{|b_i|}{3}$ elements. We have that:
	
	\[\Delta \Phi (b_i)\leq \Delta \Phi_f(b_i) + \Delta \Phi_r(b_i) \leq 
	\frac{c_0}{B} \cdot \frac{|b_i|}{6} \cdot i 
	- 2 \frac{c_0}{B} \cdot \frac{|b_i|}{6} \cdot i \leq  - \frac{c_0}{B} \cdot \frac{|b_i|}{6} \cdot i \leq 0, 
	\]
	\noindent for $|b_i|,i\geq 0$. 
\end{proof}

\begin{lemma}\label{lem:bins}
	\bi\ on an $x$-treap $D$ takes $\OO{\frac{1+\alpha}{B}\log_{\frac{\lambda}{B}} D.x}$ amortized I/Os per element, for some $\lambda \in \left[2, N \right]$ and for any real $\alpha\in (0,1]$. 
\end{lemma}

\begin{proof}
	Excluding all recursive calls (to \bi\ at Steps 1.2.2, 1.2.3 and 1.2.4 and to \init\ at Step 1.2.3), the worst-case cost of \bi\ on a buffer $b_i$ is $\scan{|b_i|^{1+\alpha}}+\OO{1}$ I/Os by Lemmata \ref{lem:ra} and \ref{lem:res}, by \cite{BFPRT73} and because this is also the worst-case I/O-cost of \fld\ and \init. The base case (Step 2) charges an extra $\OO{\frac{1}{B}}$ I/Os per element by Lemma \ref{lem:array}.
	
	In Step 1.2.2, at most $\frac{1}{3} |b_{i+1}|$ elements are moved from a buffer $b_i$ at level $i$ to a buffer $b_{i+1}$ at level $i+1$, where $|b_{i}| = |b_{i+1}|^{1+\frac{\alpha}{2}}$. The maximum difference in potential occurs when all these elements are moved from a full rear buffer of $b_i$ to a rear buffer of $b_{i+1}$. We have that:
	\[
	\sum_{j=i}^{i+1} \Delta \Phi_r (b_j) \leq 
	- 2 \frac{c_0}{B} \frac{|b_{i+1}|}{3} \left( h-i \right)
	+ 2 \frac{c_0}{B} \frac{|b_{i+1}|}{3} \left( h-i-1 \right)  \leq 
	- \frac{2}{3} \frac{c_0}{B} |b_{i+1}| \leq 0
	\]
	\noindent for $|b_{i+1}|\geq 0$.
	If the same scenario occurs between front buffers, we have that:
	\[
	\sum_{j=i}^{i+1} \Delta \Phi_f (b_j) \leq 
	- \frac{1}{3} \frac{c_0}{B} |b_{i+1}| \leq 0
	\]
	\noindent for $|b_{i+1}|\geq 0$.
	Hence $\sum_{j=i}^{i+1} \Delta \Phi (b_j)\leq 0$, given every newly inserted element is charged with an  $\OO{h}=\OO{\log_{\frac{\lambda}{B}}D.x}$ initial amount of potential.

	In Step 1.2.3, operation \spl\ removes at most $\frac{1}{2}\dx^{1+\alpha}$ from a subtreap $D$, whose bottom buffer $b_i$ we assume to be at level $i$. Operation \init\ will add these elements to a new subtreap, whose bottom buffer is also at level $i$. Without loss of generality, we focus only on these bottom buffers, since $D$ is more than half-full when \spl\ is called on it and since the bottom buffer's  size is constant fraction of the subtreap's total size. The maximum difference in potential occurs when  $\frac{1}{4}|b_{i}|$ elements are removed from a full bottom front/rear buffer of $b_i$ and added to an empty bottom front/rear buffer, respectively. We have that:
	\[
	\Delta \Phi_r (b_i) \leq 
	-  2 \frac{c_0}{B} \frac{|b_{i}|}{2} \left( h-i \right)   
	+ 2 \frac{c_0}{B} \frac{|b_{i}|}{2}  \left( h-i \right) =0
	\]
	\[
	\Delta \Phi_f (b_i)  \leq 
	-  \frac{c_0}{B} \frac{|b_{i}|}{6} \left( h-i \right)   
	-  \frac{c_0}{B}M^\varepsilon \frac{|b_{i}|}{6} \left( h-i \right)   
	\leq 0
	\]
	\noindent for $|b_{i}|\geq 0$. 
	Hence $\Delta \Phi (b_i)\leq 0$, given that subtreaps created by \init\ are charged with an $\OO{\dx^{1+\alpha} \cdot h}=\OO{\dx^{1+\alpha}\log_{\frac{\lambda}{B}}D.x}$ initial amount of potential. This extra charge does not exceed by more than a constant factor, the initial potential charged to every newly inserted element, and is amortized over the $\OO{\dx^{1+\alpha}}$ elements that are inserted to the created subtreap.
	
	In Step 1.2.4, operation \fld\ removes elements from a bottom buffer $b_i$ at level $i$ and inserts them to a middle or bottom buffer $b_{i+1}$ at level $i+1$, where $|b_{i+1}| = |b_{i}|^{1+\frac{\alpha}{2}}$. The maximum difference in potential occurs when $\frac{1}{2}|b_{i}|$ elements are removed from a full bottom rear buffer of $b_i$ and added to a rear buffer at level $i+1$. We have that:
	\[
	\sum_{j=i}^{i+1} \Delta \Phi_r (b_j) \leq 
	- 2 \frac{c_0}{B} \frac{|b_{i}|}{2} \left( h-i \right)   
	+ 2 \frac{c_0}{B} \frac{|b_{i}|}{2}  \left( h-i -1 \right) \leq
	- \frac{c_0}{B} |b_{i}| \leq 0
	\]
	\noindent for $|b_{i+1}|\geq 0$. For the case where Step \ref{fld:1} is executed, in the worst case at most $\frac{1}{6}|b_{i}|$ elements are removed from a full bottom front buffer of $b_i$. By Invariants \ref{inv:lprio} and \ref{inv:prio}, operation \re\ at Step \ref{bi:0}.1 of the recursive call to \bi\ will move these elements to the middle/bottom front buffer at level $i+1$. We have that:
	\[
	\sum_{j=i}^{i+1} \Delta \Phi_f (b_j) \leq 
	- \frac{1}{6}\frac{c_0}{B} |b_{i}| \leq 0
	\]
	\noindent for $|b_{i}|\geq 0$.
	Hence $\sum_{j=i}^{i+1} \Delta \Phi (b_j)\leq 0$.
\end{proof}

\begin{lemma}\label{lem:bext}
	\be\ on an $x$-treap $D$ takes $\OO{\lambda^{\frac{\alpha}{1+\alpha}}\frac{1+\alpha}{B}\log_{\frac{\lambda}{B}} D.x}$ amortized I/Os per element, for some $\lambda \in \left[2, N \right]$ and for any real $\alpha\in (0,1]$. 
\end{lemma}

\begin{proof}
	The cost of \be\ is dominated by the call to \flu\ on a buffer $b_i$, where $|b_i| = \OO{D.x}$. Steps \ref{be:0}.2 (extraction) and \ref{be:1} (base case) cost an extra $\OO{\lambda^{\frac{1}{1+\alpha}}/B}$ amortized I/Os per element, due to the potential's definition and by Lemma \ref{lem:array}, respectively. 
	
	Excluding all recursive calls (to \flu\ at Steps 1,  2.2 and 2.3 and to \init\ at Step 6), the worst-case cost of \flu\ on a buffer $b_i$ is $\OO{1 + \frac{|b_i|}{B}\log_{\frac{\lambda}{B}} |b_i|}$ I/Os. Specifically, \flu\ executes $\OO{|b_i|}$ \ins s and \extmin s to the temporary priority queue (that does not support \deckey) and a constant number of scans and merges. The priority queue operations \cite{ABDHM07} and Step 3 take $\OO{\frac{|b_i|}{B}\log_{\frac{\lambda}{B}} |b_i|}$ I/Os (the structure can be easily modified to achieve this bound, rather than $\sort{|b_i|}$ I/Os). This dominates the cost of incurred random accesses and of calls to \re\ (Lemmata \ref{lem:ra} and  \ref{lem:res}, respectively). The I/O-cost is amortized over the $\Theta \left(|b_i|\right)$ elements returned by \be, resulting in $\OO{\frac{1}{B}\log_{\frac{\lambda}{B}} |b_i|}$ amortized I/Os per element.
	
	To prove the negative cost of the recursive calls (Steps 1 and 2.3), it suffices to argue that there is a release in potential when buffers in consecutive levels are being processed. Without loss of generality, we assume that elements are moved from the middle front buffer at level $i$ to the top front buffer at level $i-2$ ($|b_{i}| = |b_{i-2}|^{1+\frac{\alpha}{2}}$), where the upper level subtreaps at level $i-1$ are base case structures (hence they do not affect the potential). This assumption charges an extra $\OO{\frac{\lambda^{\varepsilon}}{B}}$ I/Os per element by Lemma \ref{lem:array}. The case between the bottom and middle buffers is analogous. The maximum difference in potential occurs when $\frac{|b_{i-2}|}{4}$ elements are removed from a front middle buffer at level $i$ with less than $\frac{|b_{i}|}{4}$ elements, and added to an empty front top buffer at level $i-2$. Since the rear buffers do not change, we have that:
	\[
	\sum_{j=i-2}^{i} \Delta \Phi(b_j) \leq \sum_{j=i-2}^{i} \left(\Delta \Phi_f(b_j) +  \Delta \Phi_r(b_j)\right)\leq   \Delta \Phi_f(b_{i-2}) + \Delta \Phi_f(b_{i})\leq
	\]
	\[
	- \frac{c_0}{B} \lambda^\varepsilon   \frac{|b_i|}{4}  \left(h-i + 2\right)
	+ \frac{c_0}{B} \lambda^\varepsilon \frac{|b_i|}{4}  \left(h-i\right) \leq - \frac{c_0}{B} M^\varepsilon   \frac{|b_i|}{2} \leq 0 ,
	\]
	\noindent for $|b_i| \geq 0$.
\end{proof}

\section{Cache-oblivious priority queues}


Priority queues support operations \update\ and \extmin\ that are defined similarly to \bi\ and \be, respectively, but on a \emph{single} element. 

\subsection{Data structure}

To support these operations, we compose a priority queue out of its batched counterpart in Theorem \ref{thm:xtreap}. The data structure on $N$ elements consists of $1+\log_{1+\alpha}\log_2 N $ $x$-treaps of doubly increasing size with parameter $\alpha$ being set the same in all of them. Specifically, for $i\in \{0,\log_{1+\alpha}\log_2 N\}$, the $i$-th $x$-treap $D_i$ has $D_i.x = 2^{\left(1+\alpha \right)^i}$. We store all keys returned by \extmin\ in a hash table $X$ \cite{IP12,CFS18}. (Any I/O-efficient hash table with $\OO{1}$ query and update I/Os is also cache-oblivious). 

For $i\in \{0,\log_{1+\alpha}\log_2 N - 1\}$, we define the top buffer of $D_i$ to be ``below'' the bottom buffer of $D_{i-1}$ and the bottom buffer of $D_i$ to be ``above'' the top buffer of $D_{i+1}$. We define the set of represented pairs (key, priority) $rep = \bigcup_{i=0}^{\log_{1+\alpha}\log_2 N} D_i.rep \backslash \{(k,p)|k\in X\}$ and call \emph{represented} the keys and priorities in $rep$. We maintain the invariant that the maximum represented priority in $D_i.rep$ is smaller than the smallest represented priority below. 

\subsection{Algorithms}

To implement \update\ on a pair (key,priority) $\in rep$, we \bi\ the corresponding element to $D_0$. $D_0$ handles single-element batches, since for $i=0 \Rightarrow x=\Theta\left(1\right)$. When $D_i$ reaches capacity (i.e. contains $\left(D_i.x\right)^{1+\alpha}$ elements), we call \fld~on it, \bi\ the elements in the returned temporary array to $D_{i+1}$ and discard the array. This process terminates at the first $x$-treap that can accomodate these elements without reaching capacity.

To implement \extmin, we call \be\ to the first $x$-treap $D_i$ with a positive counter, add the extracted elements to the (empty) bottom front buffer of $D_{i-1}$ and repeat this process on $D_{i-1}$, until $D_0$ returns at least one element. If the returned key does not belong to $X$, we insert it. Else, we discard the element and repeat \extmin.

To implement \delete\ of a key, we add the key to $X$.

\begin{theorem}\label{thm:pq}
	There exist cache-oblivious priority queues on $N$ elements that support operation \update\ in $\OO{\frac{1}{B}\log_{\frac{\lambda}{B}} \frac{N}{B}}$ amortized I/Os per element and operations \extmin\ and \delete\ in $\OO{\lceil \frac{\lambda^{\frac{\alpha}{1+\alpha}}}{B}  \log_{\frac{\lambda}{B}} \frac{N}{B} \rceil \log_{\frac{\lambda}{B}} \frac{N}{B} }$ amortized I/Os per element, using $\OO{\frac{N}{B}\log_{\frac{\lambda}{B}} \frac{N}{B}}$ blocks, for some $\lambda \in \left[2, N \right]$ and for any real $\alpha\in (0,1]$. 
\end{theorem}

\begin{proof}
	To each element in $D_i$ that has been \update d since the last time that $D_i$ has undergone a \fld~operation, we define an \emph{update potential} of:
	$$ \frac{1+\alpha}{B} \log_{\frac{\lambda}{B}} \frac{2^{(1+\alpha)^i}}{B} = (1+\alpha)^{i+1}\frac{1}{B \log \frac{\lambda}{B}}-\frac{1+\alpha}{B}\log_{\frac{\lambda}{B}} B$$
	When $D_i$ reaches capacity and \fld~is called on it, the number of elements that have been \update d in the priority queue is a constant fraction of the $D_i$'s capacity. This is because the only way element could appear in $D_i$ is from:
	
	\begin{itemize}
		\item The elements that were already in the priority queue at $D_j$ for $j<i$. 
		\item The elements that were newly inserted.
		\item The elements inserted from $D_j$ for $j>i$ during \extmin\ process. However, these elements will never bring the number in $D_i$ above a constant fraction of the capacity.
	\end{itemize}
	
	\noindent Hence, an \ins\ operation, before any needed \fld\ operation, increases the update potential by:
	
	\begin{align*}
	\sum_{i=0}^{\log_{1+\alpha}\log_2 N} \left( (1+\alpha)^{i+1}\frac{1}{B \log \frac{\lambda}{B}}-\frac{1+\alpha}{B}\log_{\frac{\lambda}{B}} B \right)
	&\leq  
	\sum_{i=0}^{\log_{1+\alpha}\log_2 N} (1+\alpha)^{i+1}\frac{1}{B \log \frac{\lambda}{B}}
	\\
	&= O \left( \frac{\log N}{B \log \frac{\lambda}{B}} \right)
	\\
	&= O \left( \frac{\log_\frac{\lambda}{B} N}{B} \right)
	\end{align*}

	To every key $k$ that has been returned by operation \extmin\ and occurs in $d$ elements in $D$, we define a potential of $\Phi(k) = d\cdot \frac{c_d}{B}\lambda^{\varepsilon}\cdot \log_{\frac{\lambda}{B}} D.x $ (for constant $c_d\geq 1$). \ins\ does not affect this potential by the assumption that an extracted key is not reinserted to the structure. 
	
	To $D_i$ we give an \emph{extract-min potential} of:
	$$ \lambda^{\frac{\alpha}{1+\alpha}}\frac{1+\alpha}{B} \log_{\frac{\lambda}{B}} \frac{2^{(1+\alpha)^i}}{B} 
	= (1+\alpha)^{i+1}\frac{\lambda^{\frac{\alpha}{1+\alpha}}}{B \log \frac{\lambda}{B}}-\lambda^{\frac{\alpha}{1+\alpha}}\frac{1+\alpha}{B}\log_{\frac{\lambda}{B}} B$$ 
	for the number of \extmin~operations performed since the last time a \be\ operation was performed on $D_i$. When $D_i$ gets empty and \extmin~is needed to fill it, the number of elements \extmin 'd from the priority queue since the last time that happened is a constant fraction of $D_i$'s capacity. This is true, since the only way elements could be removed from $D_i$ between \be s is because $D_i$ reached its capacity from \update s and thus \fld\ needed to be called on $D_{i-1}$. This, however, only reduces the number of elements to a constant fraction of $D_i$'s capacity.
	
	\noindent Hence, an \extmin~operation, before any needed \fld, increases the extract-min potential by:
	
	\begin{align*}
	\sum_{i=0}^{\log_{1+\alpha}\log_2 N} \left( (1+\alpha)^{i+1}\frac{\lambda^{\frac{\alpha}{1+\alpha}}}{B \log \frac{\lambda}{B}}-\lambda^{\frac{\alpha}{1+\alpha}}\frac{1+\alpha}{B}\log_{\frac{\lambda}{B}} B \right)
	&\leq  
	\sum_{i=0}^{\log_{1+\alpha}\log_2 N} (1+\alpha)^{i+1}\frac{\lambda^{\frac{\alpha}{1+\alpha}}}{B \log \frac{\lambda}{B}}
	\\
	&= O \left( \lambda^{\frac{\alpha}{1+\alpha}}\frac{\log N}{B \log \frac{\lambda}{B}} \right)
	\\
	&= O \left( \lambda^{\frac{\alpha}{1+\alpha}}\frac{\log_{\frac{\lambda}{B}} N}{B} \right)
	\end{align*}

	\noindent Finally, the potential is constructed to exactly pay for the cost of \fld~between $x$-treaps. 
	
	The worst-case cost of operation \delete\ is $\OO{1}$ I/Os \cite{IP12}. However, its amortized cost includes the potential necessary to remove all the elements with the deleted key from the structure. We introduce an extra ``ghost potential'' of $\OO{\frac{\lambda^{\frac{\alpha}{1+\alpha}}}{B}\log_{\frac{\lambda}{B}} \frac{N}{B}}$ to every key in $X$ that has been \delete d or \extmin 'd. The stated I/O-cost for \delete\ follows from the fact that the first time a key is \delete d or \extmin 'd from the structure $\OO{\log_{\frac{\lambda}{B}}\NN}$ elements in the structure get charged by this potential (by Lemma \ref{lem:card}). Whenever an \extmin\ returns an element with ghost potential, this potential is released in order to \extmin\ the next element at no extra amortized cost.
	
	The priority queue may contain an $N$-treap that occupies $\OO{\left(\NN\right)^{1+\alpha}\log_{\frac{\lambda}{B}}\NN}$ blocks (by Theorem \ref{thm:xtreap}). However, with carefull manipulation of the unaddressed empty space, this space usage (and thus of the whole priority queue) can be reduced back to $\OO{\NN\log_{\frac{\lambda}{B}}\NN}$ blocks.
\end{proof}

\section{Cache-oblivious buffered repository trees}

A \emph{buffered repository tree (BRT)} \cite{BGVW00,ABDHM07,CR18} stores a multi-set of at most $N$ \emph{elements}, each associated with a \emph{key} in the range $\left[1\dots k_{\max}\right]$. It supports the operations \ins~and \ext~that, respectively, insert a new element to the structure and remove and report all elements in the structure with a given key. 
To implement a BRT, we make use of the $x$-box \cite{BDFILM10}.  
Given positive real $\alpha \leq 1$ and key range $\left[k_{\min}, k_{\max}\right)\subseteq \Re$, an $x$-\emph{box} $D$ stores a set of at most $\frac{1}{2}\left(D.x\right)^{1+\alpha} $ elements associated with a key $k\in \left[D.k_{\min}, D.k_{\max}\right)$. An $x$-box supports the following operations: 

\begin{itemize}
	
	\item \bi $\left(D, e_1, e_2 , \ldots, e_{b}\right)$: For constant $c\in \left(0,\frac{1}{2}\right]$, insert $b \leq c\cdot D.x$ elements $e_1, e_2 , \ldots, e_{b}$  to $D$, given they are key-sorted with keys  $e_i.k \in \left[D.k_{\min}, D.k_{\max}\right), i\in \left[1,b\right]$. 
	
	
	\item \sea $\left(D, \kappa \right)$: Return pointers to all elements in $D$ with key $\kappa$, given they exist in $D$ and $\kappa \in \left[D.k_{\min}, D.k_{\max}\right)$.
	
\end{itemize}

To implement operation \ext$\left(D, \kappa \right)$~that extracts all elements with key $\kappa$ from an $x$-box $D$, we \sea$\left(D, \kappa \right)$ and remove from $D$ all returned pointed elements. 

The BRT on $N$ elements consists of $1+\log_{1+\alpha}\log_{2} N$ $x$-boxes of doubly increasing size with parameter $\alpha$ being set the same in all of them. We obtain the stated bounds by modifying the proof of the $x$-box \cite[Theorem 5.1]{BDFILM10} to account for Lemmata \ref{lem:xbox} and \ref{lem:brtarray}.

\begin{lemma}\label{lem:xbox}
	For $D.x= \Omega \left( \lambda^{\frac{1}{1+\alpha}}\right)$, an $x$-box supports operations \bi\ in  $\OO{\frac{1+\alpha}{B}\log_{\frac{\lambda}{B}} \frac{D.x}{B}}$ and \ext~on $K$ extracted elements in $\OO{\left(1+\alpha \right)\log_{\frac{\lambda}{B}} \frac{D.x}{B} +\frac{K}{B}}$ amortized I/Os per element, using $\OO{\frac{\dx^{1+\alpha}}{B}}$ blocks, for some $\lambda \in \left[2, N \right]$ and for any real $\alpha\in (0,1]$. 
\end{lemma}

\begin{proof}
	Regarding \bi\ on a cache-aware $x$-box, we obtain $\OO{\frac{1+\alpha}{B}\log_{\frac{\lambda}{B}} \frac{D.x}{B}}$ amortized I/Os by modifying the proof of \bi\ \cite[Theorem 4.1]{BDFILM10} according the proof of Lemma \ref{lem:bins}. Specifically, every element is charged $\OO{1/B}$ amortized I/Os, instead of $\OO{1/B^{\frac{1}{1+\alpha}}}$, and the recursion stops when $D.x = \OO{\lambda^{\frac{1}{1+\alpha}}}$, instead of $D.x = \OO{B^{\frac{1}{1+\alpha}}}$.
	
	Regarding \sea ing for the first occurrence of a key in a cache-aware $x$-box, we obtain $\OO{\log_{\frac{\lambda}{B}} \frac{D.x}{B}}$ amortized I/Os by modifying the proof of \sea~\cite[Lemma 4.1]{BDFILM10}, such that the recursion stops when $D.x = \OO{\lambda^{\frac{1}{1+\alpha}}}$, instead of $D.x = \OO{B^{\frac{1}{1+\alpha}}}$. To \ext\ all $K$ occurrences of the searched key, we access them by scanning the $x$-box and by following fractional cascading pointers, which incurs an extra $\OO{\frac{K}{B}}$ I/Os. 
\end{proof}

\begin{lemma}\label{lem:brtarray}
	An $\OO{\lambda^{\frac{1}{1+\alpha}}}$-box supports operation \bi\ in $\OO{1/B}$ amortized I/Os per element and operation \ext~on $K$ extracted elements in $\scan{\lambda^{\frac{\alpha}{1+\alpha}}}$ amortized I/Os per element, for some $\lambda \in \left[2, N \right]$ and for any real $\alpha\in (0,1]$. 
\end{lemma}

\begin{proof}
	We allocate an array of size $\OO{M}$ and implement \bi\ by simply appending the inserted element to the array and \ext~by scanning the array and removing and returning all occurrences of the searched key. 
\end{proof}

\begin{theorem}\label{thm:brt}
	There exist cache-oblivious buffered priority trees on a multi-set of $N$ elements and  $K$ extracted elements that support operations \ins\ in $\OO{\frac{1}{B}\log_{\frac{\lambda}{B}} \frac{N}{B}}$ and \ext\ in $\OO{\frac{\lambda^{\frac{\alpha}{1+\alpha}}}{B}\log_{\frac{\lambda}{B}} \frac{N}{B}+\frac{K}{B}}$ amortized I/Os per element, using $\OO{\frac{N}{B}}$ blocks, for some $\lambda \in \left[2, N \right]$ and for any real $\alpha\in (0,1]$. 
\end{theorem}

\section{Applications to graph algorithms}

\subsection{Directed single-source shortest paths}

\begin{theorem} \label{thm:dsssp}
	Single source shortest paths on a graph with $V$ nodes and directed $E$ edges can be computed cache-obliviously in $\OO{\frac{V^{\frac{1}{1+\alpha}} E^{\frac{\alpha}{1+\alpha}}}{B}\log^2_{\frac{E}{VB}} \frac{E}{B} + V\log_{\frac{E}{VB}} \frac{E}{B} +  \frac{E}{B} \log_{\frac{E}{VB}} \frac{E}{B}}$ I/Os, for any real $\alpha\in (0,1]$. 
\end{theorem}

\begin{proof}
	The algorithm of Vitter \cite{V01} (described in detail in \cite[Lemma 4.1]{CR18} for the cache-oblivious model) makes use of a priority queue that supports the \update\ operation and of a BRT on $\OO{E}$ elements. Specifically, it makes $V$ calls to \extmin\ and $E$ calls to \update\ on the priority queue and $V$ calls to \ext\ and $E$ calls to \ins\ on the BRT. We obtain $\OO{V \frac{\lambda^{\frac{\alpha}{1+\alpha}}}{B}\log^2_{\frac{\lambda}{B}} \frac{E}{B} + V\log_{\frac{\lambda}{B}} \frac{E}{B} +  \frac{E}{B} \log_{\frac{\lambda}{B}} \frac{E}{B}}$ total I/Os by using Theorems \ref{thm:pq} and \ref{thm:brt} for some $\lambda \in \left[2, N \right]$. We set $\lambda = \OO{E/V}$ to obtain the stated bound. 
\end{proof}

\subsection{Directed depth- and breadth-first search}

\begin{theorem} \label{thm:ddbfs}
	Depth-first search and breadth-first search numbers can be assigned cache-obliviously to the $V$ nodes of a graph with $E$ directed edges in $\OO{\frac{V^{\frac{1}{1+\alpha}} E^{\frac{\alpha}{1+\alpha}}}{B}\log^2_{\frac{E}{VB}} \frac{E}{B} + V\log_{\frac{E}{VB}} \frac{E}{B} +  \frac{E}{B} \log_{\frac{E}{VB}} \frac{E}{B}}$ I/Os, for any real $\alpha\in (0,1]$. 
\end{theorem}

\begin{proof}
	The algorithm of Buchsbaum et al. \cite{BGVW00} makes use of a priority queue and of a BRT on $\OO{E}$ elements. Specifically, it makes $2V$ calls to \extmin\ and $E$ calls to \ins\ on the priority queue and $2V$ calls to \ext~and $E$ calls to \ins\ on the BRT \cite[Theorem 3.1]{BGVW00}. We obtain $\OO{V \frac{\lambda^{\frac{\alpha}{1+\alpha}}}{B}\log^2_{\frac{\lambda}{B}} \frac{E}{B} + V\log_{\frac{\lambda}{B}} \frac{E}{B} +  \frac{E}{B} \log_{\frac{\lambda}{B}} \frac{E}{B}}$ total I/Os, by using Theorems \ref{thm:pq} and \ref{thm:brt} for some $\lambda \in \left[2, N \right]$. We set $\lambda = \OO{E/V}$ to obtain the stated bound. 
\end{proof}

\bibliographystyle{plain}
\bibliography{xtreap}

\end{document}